\documentclass[12pt,a4paper]{article}
\usepackage[truedimen,margin=30mm]{geometry} 

\usepackage{mathrsfs}
\usepackage{amssymb}
\usepackage{amsmath}
\usepackage{ascmac}
\usepackage{amsthm}
\usepackage[pdftex]{graphicx}
\usepackage{natbib}
\usepackage{setspace}
\usepackage{url}

\usepackage{color}
\usepackage{times}   %
\usepackage{comment}
\usepackage{bm}
\usepackage{titlesec}
\titleformat*{\section}{\large\bfseries}
\titleformat*{\subsection}{\it}
\titleformat*{\subsubsection}{\it}

\newtheorem{thm}{Theorem}
\newtheorem{lem}{Lemma}

\newtheorem{remark}{Remark}

\def\ep{{\varepsilon}}
\def\la{{\lambda}}
\def\si{{\sigma}}

\def\Ga{\Gamma}

\def\la{{\lambda}}
\def\non{{\nonumber}}

\def\alt{{\tilde \al}}
\def\bet{{\tilde \be}}
\def\etat{{\tilde \eta}}

\def\pit{{\tilde \pi}}
\def\gat{{\tilde \ga}}

\def\Kc{{\cal K}}
\def\Lc{{\cal L}}

\def\pit{{\tilde \pi}}

\def\pd{\partial}

\def\al{{\alpha}}
\def\be{{\beta}}
\def\ga{{\gamma}}
\def\de{{\delta}}
\def\ep{{\varepsilon}}
\def\la{{\lambda}}
\def\si{{\sigma}}
\def\om{{\omega}}

\def\ka{{\kappa}}

\def\bbe{{\text{\boldmath $\beta$}}}

\def\bphi{{\text{\boldmath $\phi$}}}

\def\bmu{{\text{\boldmath $\mu$}}}

\def\alt{{\tilde \al}}
\def\bet{{\tilde \be}}
\def\etat{{\tilde \eta}}

\def\bbet{{\widetilde \bbe}}

\def\De{{\Delta}}

\def\Ga{{\Gamma}}

\def\bPsi{{\text{\boldmath $\Psi$}}}

\def\b{{\text{\boldmath $b$}}}

\def\e{{\text{\boldmath $e$}}}

\def\x{{\text{\boldmath $x$}}}

\def\A{{\text{\boldmath $A$}}}

\def\O{{\text{\boldmath $O$}}}

\def\Jh{{\widehat J}}
\title{{\bf 
Robust Bayesian Inference for Censored Survival Models
}}
\date{}

\begin{document}

\maketitle
\doublespacing

\vspace{-1.5cm}
\begin{center}
{\large 
Yasuyuki Hamura$^{1}$, Takahiro Onizuka$^{2}$, Shintaro Hashimoto$^{3}$ and Shonosuke Sugasawa$^{4}$
}
\end{center}

\medskip
\noindent
$^{1}$Graduate School of Economics, Kyoto University\\
$^{2}$Graduate School of Social Sciences, Chiba University\\
$^{3}$Department of Mathematics, Hiroshima University\\
$^{4}$Faculty of Economics, Keio University

\medskip
\medskip
\medskip
\begin{center}
{\bf \large Abstract}
\end{center}

This paper proposes a robust Bayesian accelerated failure time model for censored survival data. We develop a new family of life-time distributions using a scale mixture of the generalized gamma distributions, where we propose a novel super heavy-tailed distribution as a mixing density. 
We theoretically show that, under some conditions, the proposed method satisfies the full posterior robustness, which guarantees robustness of point estimation as well as  uncertainty quantification. 
For posterior computation, we employ an integral expression of the proposed heavy-tailed distribution to develop an efficient posterior computation algorithm based on the Markov chain Monte Carlo. 
The performance of the proposed method is illustrated through numerical experiments and real data example. 

\vspace{-0cm}


\newpage
%
\section{Introduction}
\label{sec:1}

Censored survival data arise in various fields such as medical research, reliability engineering, and economics. When covariate information is available, the Cox proportional hazards model \citep{cox1972regression} is a widely used semi-parametric approach for modeling survival outcomes. However, unobserved heterogeneity related to influential observations can lead to biased parameter inference. For instance, the proportional hazards assumption may be violated in the presence of unobserved confounders \citep{Omori1993}.
As an alternative, the accelerated failure time (AFT) model provides a parametric framework for analyzing censored survival data. 
The AFT model does not rely on the proportional hazards assumption and can be fitted by the maximum likelihood method, enabling a complete description of the hazard function. 
However, standard AFT models often assume light-tailed distributions such as the gamma, Weibull, or log-normal, making parameter estimation sensitive to outliers \citep[e.g.][]{klein2014handbook}.

For robust parameter estimation in AFT models, it is useful to consider infinite mixtures of lifetime distributions by introducing local parameters or random effects \citep[e.g.][]{Marshall2007}. 
Several Bayesian approaches have been proposed in this context. For example, \cite{VallejosSteel2015} introduced a shape mixture of log-normal distributions, \cite{VallejosSteel2017} adopted a rate mixture of Weibull distributions, and \cite{Lachos2017} employed a scale mixture of log-Birnbaum-Saunders distributions.
For outlier detection, \cite{VallejosSteel2015} proposed a method based on the Bayes factor. However, quantifying the uncertainty associated with the presence or absence of outliers remains challenging. Moreover, despite the variety of Bayesian approaches available, a theoretically established framework for ensuring robustness in AFT models has yet to be developed.

To address the limitations of existing robust models for censored survival data, we propose a novel approach by introducing a local parameter into the AFT model, where the local parameter follows an extremely heavy-tailed distribution. Specifically, we focus on the AFT model based on the generalized gamma distribution, a flexible family that includes typical lifetime distributions such as the exponential, gamma, Weibull, and log-normal distributions as special cases \citep[see, e.g.,][]{stacy1962generalization, Marshall2007}. In the proposed model, the local parameter follows a two-component mixture density with a spike-and-slab structure, where the slab component is modeled using a heavy-tailed loglog-Pareto-like distribution. For Bayesian inference, we develop an efficient posterior computation algorithm using data augmentation and appropriate reparameterization. Notably, the two-component structure allows us to quantify the uncertainty regarding the presence of outliers for each observation through a latent binary variable.

A key theoretical contribution of this work is establishing rigorous posterior robustness for the proposed model, where information from outliers is automatically downweighted in the posterior distribution when the outliers are sufficiently large \citep[e.g.,][]{Gagnon2020, Hamura2022log}. While posterior robustness has been extensively studied in the context of linear regression models \citep{Gagnon2020, gagnon2023theoretical, Hamura2022log, hamura2024posterior} and, more recently, in count regression models \citep{hamura2024robust}, no theoretical results have been established for censored survival outcomes. We demonstrate that the posterior distribution of the proposed model satisfies posterior robustness under mild conditions, filling this gap in the literature. Therefore, the proposed method is appealing not only for its computational efficiency but also for its theoretical justification.

The remainder of the paper is structured as follows. In Section \ref{sec:2}, we introduce the generalized gamma distribution with a local scale parameter and propose an accelerated failure time model based on the scale mixtures of the generalized gamma distributions. The posterior computation algorithm of the proposed method is also provided. In Section \ref{sec:3}, the main theorem of this paper is presented. We show that the proposed model satisfies the property of theoretical robustness under some conditions. In Section \ref{sec:sim}, we illustrate the performance of the proposed method through numerical experiments. In Section \ref{sec:analysis}, we apply the proposed methods to real data analysis. 
R code implementing the proposed methods is available at Github repository (\url{https://github.com/sshonosuke/robust_survival})

\section{Robust Censored Survival Models }
\label{sec:2}

\subsection{Survival models with generalized gamma distributions}
Let $y_1,\ldots,y_n$ be the censored survival time.
Furthermore, let $\delta_1,\ldots,\delta_n$ be a censoring indicator such that $\delta_i=1$ for $i=1,\ldots,n$, when the observed survival time $y_i$ is not censored and $\delta_i=0$ otherwise. 
When $\delta_i=0$, there exists an unobserved survival time $t_i$ such that $y_i=\min(C_i, t_i)$, where $C_i$ is a censoring time.  
If $\de _i = 1$, then $y_i = t_i$. 
For $t_i$, we consider the following generalized gamma (GG) distribution: 
\begin{equation}\label{GG}
p_{\mathrm{GG}}(t_i| \alpha , \gamma , \theta_i ) 
= \frac{\gamma\alpha^\alpha}{\Gamma(\alpha)}\theta_i^\alpha t_i^{\alpha\gamma-1}\exp(-\alpha\theta_i t_i^\gamma), \ \ \ \ \ \ t_i>0.
\end{equation}
The GG distribution is recognized as a flexible family of distribution for positive valued data, and has been adopted in modeling survival outcomes. 
For example, \cite{cox2007parametric} provided a useful classification of the hazard function for a kind of generalized gamma distribution and \cite{shukla2023bayes} proposed a Bayesian AFT model based on the generalized gamma distribution.
In (\ref{GG}), $\theta_i$ is an individual-specific parameter controlling scale of the distribution, and $\alpha$ and $\gamma$ are two shape parameters. 
Representative special cases of the GG distribution are gamma ($\gamma=1$) and Weibull ($\alpha=1$) distributions.

Given an auxiliary information $\x_i$, we assume $\theta_i= \exp( \x_i^\top \bbe )/\lambda_i$, where $\lambda_i$ is a local parameter representing an individual-specific effect and $\bbe $ is an unknown vector of regression coefficients. 
Note that the expectation of the model (\ref{GG}) given $\lambda_i$ is given by
\begin{align}\label{expectation-gg}
E[t_i|\lambda_i] = (\alpha \theta_i)^{-1/\gamma} \frac{\Gamma\left(\alpha+1/\gamma\right)}{\Gamma(\alpha)}=\left( \frac{\exp( \x_i^\top \bbe )}{\lambda_i} \right)^{-1/\gamma} \alpha^{-1/\gamma} \frac{\Gamma\left(\alpha+1/\gamma\right)}{\Gamma(\alpha)},
\end{align}
and the reliability function $R_i(t|\theta_i,\alpha,\gamma)$ can be expressed as 
\begin{align}\label{reliability}
R_i(t|\theta_i,\alpha,\gamma) = \int_t^{\infty} \frac{\gamma \alpha^{\alpha}}{\Gamma(\alpha)} \theta_i^{\alpha} x^{\alpha \gamma -1} \exp(-\alpha \theta_i x^{\gamma}) dx = \frac{\Gamma_I(\alpha; \alpha \theta_i t^{\gamma})}{\Gamma(\alpha)},
\end{align}
where $\Gamma_I(\cdot; \cdot)$ denoted the upper incomplete gamma function. The corresponding hazard rate function $h_i(t|\lambda_i, \x _i )$ can be calculated by dividing \eqref{GG} by \eqref{reliability}.

The model (\ref{GG}) includes the gamma model ($\gamma=1$) and the Weibull model ($\alpha=1$), described as follows:  

\begin{itemize}
\item[-]{\bf (Gamma model)} When $\gamma=1$, the model (\ref{GG}) is equivalent to $t_i |\theta_i\sim {\rm Ga}(\phi, \phi\theta_i)$ with $\theta_i=\exp(-\x_i^\top \bbe)/\lambda_i$.
The mean and variance are 
\[
E[t_i|\la_i]=\la_i\exp( \x_i^\top \bbe ) \quad \text{and} \quad  {\rm Var}(t_i|\lambda_i)=\la_i^2\exp(2 {\x _i} ^\top \bbe )/\phi.
\]
Hence, $\lambda_i$ adjusts the mean and variance of each individual and, particularly, $\lambda_i$ can be large for an outlying individual having a long survival time that cannot be well explained by $\exp( \x_i^\top \bbe )$.

\item[-]{\bf (Weibull model)} 
When $\alpha=1$, the model (\ref{GG}) reduces to the Weibull model, described as 
\[
p(t_i|\theta,\phi)=\theta_i\phi t_i^{\phi-1}\exp(-\theta_it_i^{\phi}),
\]
with $\theta_i= \exp( \x_i^\top \bbe )/\lambda_i$.
The hazard function under the Weibull model is given by 
\[
h_i(t|\la_i, \x _i )=\frac{\exp( \x_i^\top \bbe )}{\la_i}\phi t^{\phi}.
\]
Here $\phi t^{\phi}$ is a baseline hazard function, and the other term reflects the potential heterogeneity of the individual hazard. 
When $t_i$ is an outlier having large survival time compared to $\x _i$, $\lambda_i$ can be large so that the resulting hazard is extremely small for such an individual.  
\end{itemize}

For the gamma and Weibull models, the interpretation of $\lambda_i$ is clear and represents the part of the survival time of each observation that cannot be explained by regression. If $\lambda_i$ is estimated to be 1, the survival time of individual $i$ can be well explained by the regression part as well as the standard AFT model. Many existing methods are known to be difficult to interpret as such. For example, the AFT model based on log-normal distribution by \cite{VallejosSteel2015} is defined by
\begin{align}\label{steel}
\log t_i |\lambda_i\sim \mathrm{N}(\x_i^{\top} \bbe , \sigma^2/\lambda_i), \ \ \ \ \ \lambda_i\sim  \pi(\lambda_i),
\end{align}
where $\pi(\lambda_i)$ is a mixing density. Under the model, the conditional expectation of the survival time $t_i$ given $\lambda_i$ is expressed as 
\begin{align}\label{expectation-LN}
E[t_i|\lambda_i, \bbe ,\sigma] =\exp\left( \x_i^{\top} \bbe + \frac{\sigma^2}{2 \lambda_i}\right),
\end{align}
while the role of the parameter $\lambda_i$ is unclear compared to the gamma and Weibull models because the parameter $\sigma^2$ must be estimated. The comparison between the proposed model and the model \eqref{steel} is discussed in Sections \ref{sec:analysis} and \ref{subsec:supp-logT} of the Supplementary Material. The ease of interpreting local parameters is another advantage of considering the generalized gamma distribution and its submodels.

In either case, the prior distribution of $\lambda_i$ should have a heavy tail to handle outlying observations, which makes robust inference on the other parameters such as $\bbe $. In this paper, we consider a family of scale mixture of generalized gamma distributions defined by
\begin{equation*}
f(t_i|\alpha,\bbe, \gamma) = \int_0^{\infty} p_{\mathrm{GG}}(t_i|\alpha,\gamma,\exp( \x_i^\top \bbe )/\lambda_i) \pi(\lambda_i)d \lambda_i,
\end{equation*}
where $p_{\mathrm{GG}}(t_i|\alpha,\gamma,\theta_i)$ corresponds to the density of the generalized gamma distribution defined by \eqref{GG}.

\subsection{Robust Bayesian inference via scale mixture of GG distributions}
\label{subsec:22} 

We consider the following two-component mixture local prior for $\lambda_i$: 
\begin{equation}\label{eq:local_double_log_right} 
\pi(\lambda_i | s, c)=(1-s) \cdot \delta_1 (\lambda_i) + s \cdot G_{\mathrm{DLH}}(\lambda_i | c) , 
\end{equation}
where $s \in (0, 1)$ is the probability of observing an outlying observation. 
The first component $\delta_1$ is point mass at $1$ such that the posterior distribution may be close to what would have been obtained under the usual model in the absence of outliers ($s \to 0$). 
The second component $G_{\mathrm{DLH}}$ is the doubly log-adjusted heavy-tailed (DLH) distribution with a fixed hyperparameter $c > 0$, described as 
\begin{align}
G_{\mathrm{DLH}}(\lambda_i | c) 
={c \over 1 + \la _i} {1 \over 1 + \log (1 + \la _i )} {1 \over [1 + \log \{ 1 + \log (1 + \la _i ) \} ]^{1 + c}},
\label{eq:DLH_2025} 
\end{align}
for $\la_i >0$.
The distribution (\ref{eq:DLH_2025}) is a special case of the iteratively log-adjusted distribution considered in \cite{hamura2020shrinkage}, and can be interpreted as a loglog-Pareto distribution.
A notable property is that it has an extremely heavy tail, that is, $G_{\mathrm{DLH}}(\lambda_i | c)\approx \lambda_i^{-1}$ ignoring log-factors. 
This is a key feature for accommodating outliers to achieve theoretically valid robust Bayesian inference. 
Incidentally, as seen in Section \ref{subsec:supp-nonrobust} of the Supplementary Material, the use of the iterated logarithm is important for robustness in our general setting. 
This is in contrast, for example, to the existing theoretical results on the posterior robustness under linear regression \citep{Hamura2022log} and count data \citep{hamura2024robust}, where only one log-adjustment term is included.

We introduce a latent binary variable $z_i\in \{0,1\}$ following a Bernoulli distribution, $z_i\sim {\rm Ber}(s)$.
Then, the mixture distribution (\ref{eq:local_double_log_right}) can be expressed as 
$$
\lambda_i|(z_i=1) \sim G_{\mathrm{DLH}}, \ \ \ \lambda_i|(z_i=0) \sim \delta_1.
$$
When the $i$th observation is not outlier, we can fit the standard GG model to the data, so that $\lambda_i=1$ (i.e. $z_i=0$). 
On the other hand, when $i$th observation is an outlier, $z_i$ should be $1$ and introduce heavy-tailed distribution for $\lambda_i$ to absorb the effect of outliers. 
Hence, the latent variable $z_i$ can be regarded as an indicator of whether the $i$th observation is an outlier or not, so that by computing the posterior probability of $z_i$, one can obtain the posterior probability of being outliers for each observation. 
Based on the above representation with $z_i$, we can also express $\lambda_i$ as $\lambda_i=\eta_i^{z_i}$ and $\eta_i\sim G_{\mathrm{DLH}}$, which will be used in developing a posterior computation algorithm.

\subsection{Prior and posterior distributions}
We perform Bayesian inference by assigning prior distributions for unknown parameters. 
The proposed model includes several model parameters.
The main GG distribution includes $\bbe $ (regression coefficients) and $(\alpha,\gamma)$ (two shape parameters in GG), and the local prior of $\lambda_i$ includes $s$ (mixing proportion). 
For these parameters, we employ the priors as $\bbe\sim {\rm{N}}_p ( \b _{\bbe } , \A_{\bbe}^{-1} )$, $\alpha\sim {\rm GIG} ( a_{\al } , b_{\al } , c_{\al } )$, $\gamma\sim {\rm{GIG}} ( a_{\ga } , b_{\ga } , c_{\ga } )$, and $s\sim {\rm{Beta}} ( a_s , b_s )$, where $\b _{\bbe } ,\A_{\bbe}^{-1}, a_{\al } , b_{\al } , c_{\al } , a_{\ga } , b_{\ga }, c_{\ga}$, $a_s$ and $b_s$ are fixed hyperparameters and ${\rm GIG}( a_0 , b_0 , c_0 )$ denotes the generalized inverse Gaussian (GIG) distribution with parameters $a_0 , b_0$ and $c_0$ whose density $p_{\rm GIG}(x| a_0 , b_0 , c_0 )$ is 
$$
p_{\rm GIG}(x| a_0 , b_0 , c_0 ) \propto x ^{c_0 - 1} \exp (- a_0 x - b_0 / x ).
$$
Note that $\A_{\bbe}^{-1}$ should be positive definite, and scalar parameters, $a_{\al } , b_{\al } , a_{\ga } , b_{\ga }$, $a_s$ and $b_s$ should be positive. 
A prior distribution for $( \eta _i , z_i )$ consistent with (\ref{eq:local_double_log_right}) is given by $p( \eta_i , z_i ) = G_{\mathrm{DLH}}(\eta_i;  c)  \times s^{z_i} (1 - s)^{1 - z_i}$.

Let $Z=\{z_1,\ldots,z_n\}$ and $E=\{\eta_1,\ldots,\eta_n\}$ be collections of latent variables. 
Also, we define $\Phi=\{\bbe , \alpha, \gamma, s\}$ as a collection of unknown model parameters. 
Under the prior settings described above, the joint posterior distribution of $Z, E$ and $\Phi$ given the survival time, $\mathcal{T}=\{t_1,\ldots,t_n\}$ is given by
\begin{equation}\label{eq:pos}
\begin{split}
&p(Z, E, \Phi \mid \mathcal{T}) \\
&\propto \frac{\al^{c_{\al} + n\alpha -1}}{\Gamma(\alpha)^n} \ga^{c_{\ga} + n - 1}
\exp\left\{ - a_{\al } \al - \frac{b_{\al}}{\al} - a_{\ga } \ga - \frac{b_{\ga }}{\ga}  -\frac12 ( \bbe - \b_{\bbe} )^\top \A_{\bbe} ( \bbe - \b_{\bbe} ) \right\} \\
& \ \ \ \ \ \ \ 
\times  s^{a_s - 1} (1 - s)^{b_s - 1} \prod_{i = 1}^{n}  s^{z_i} (1 - s)^{1 - z_i}\frac{t_i^{\ga \al}}{\eta_i^{\alpha z_i}}
\exp\left\{ \alpha\x_i^{\top} \bbe  - \al \exp ( \x_i^{\top} \bbe ) \frac{{t_i}^{\ga }}{{\eta _i}^{z_i}} \right\} \\
& \ \ \ \ \ \ \ 
\times G_{\rm DLH}( \eta _i | c) \text{,}
\end{split}
\end{equation}
where $G_{\rm DLH}( \eta _i | c)$ is given by (\ref{eq:DLH_2025}). 
Note that the above posterior is based on the uncensored survival time $\mathcal{T}$, not on the censored survival time $\mathcal{Y}=\{y_1,\ldots,y_n\}$.
Hence, the survival time $t_i$ under $\delta_i=0$ is missing and is treated as an unobserved latent variable. 
As described in the next section, unobserved survival time can be easily imputed within Markov Chain Monte Carlo iterations.

\subsection{Posterior computation algorithm}\label{subsec:pos} 

Here, we provide an efficient posterior computation algorithm of the joint posterior (\ref{eq:pos}).
Given the latent variables, $\eta_i$ and $z_i$, the posterior computation of (\ref{eq:pos}) reduces to exploring the posterior of the GG model. 
For the latent variables, the full conditional distribution of $\eta_i$ is not a familiar form, but we provide an efficient Gibbs sampler by employing a novel integral expression of the DLH distribution (\ref{eq:DLH_2025}). 
Specifically, the DLH distribution (\ref{eq:DLH_2025}) holds the following integral expression:  
\begin{align}
G_{\rm DLH}(\eta | c)
&= \int_{(0, \infty )^3} {u^{1 + c - 1} e^{- u} \over \Ga (1 + c)} {v^{1 + u - 1} e^{- v} \over \Ga (1 + u)} {c \over {\eta }^{1 + v}} {w^{1 + v - 1} e^{- w} \over \Ga (1 + v)} e^{- w / \eta } d(u, v, w) \non 
\end{align}
for all $\eta \in (0, \infty )$.
Then, we can introduce additional latent variables $u_i , v_i , w_i \in (0, \infty )$ so that the full conditional of $\eta _i$ becomes a familiar one, as shown below. 
When $z_i = 0$, we sample $\eta _i$ directly from the DLH distribution.

In order to decrease autocorrelation, we employ reparameterization of the model. 
Instead of $(\alpha, \beta)$, we consider $ \alt=\alpha \gamma^2$ and $\bbet=\bbe/\gamma$ and then use the independent Metropolis-Hastings algorithm based on the ideas of \cite{Miller2019}. 
We also consider the change of variable $\eta _i = {\tilde{\eta } _i}^{1 - z_i + z_i \ga }$ and  $\gat = \ga / (1 + \ga ) \in [0, 1]$ and use the piecewise linear approximation \citep[e.g.][Section 3.1.4]{luengo2020survey} for $\gat$.

Details of the posterior computation algorithm are described as follows: 

\begin{itemize}
\item
(Sampling of censored $t_i$) \ For $i = 1, \dots , n$ with $\de_i = 0$, generate $t_i$ from the truncated GG distribution proportional to 
\begin{align*}
I( t_i > C_i ) {t_i}^{\ga \al - 1} \exp\left\{- \al \la_i^{-1}\exp ( \x_i^{\top} \bbe) {t_i}^{\ga} \right\}.
\end{align*}

\item
(Sampling of $z_i$) \ For $i = 1, \dots , n$, generate the binary indicator $z_i \in \{ 0, 1 \} $ from the Bernoulli distribution with success probability being proportional to 
\begin{align*}
s^{z_i} (1 - s)^{1 - z_i} (1 - z_i + z_i \ga ) {\tilde{\eta } _i}^{ z_i(\ga-1) } G_{\rm DLH}({\tilde{\eta}_i}^{1 - z_i + z_i \ga} | c) 
p_{\rm Ga}
(\al t_i^{\gamma}| \alpha, \exp ( {\x _i}^{\top } \bbe ) / {\tilde{\eta } _i}^{z_i \ga } ),
\end{align*}
where $p_{\rm Ga}(x|a,b)$ denotes the gamma density with shape parameter $a$ and rate parameter $b$.  

\item
(Sampling of $\eta_i$) \ For $i=1,\ldots,n$, first generate the three latent variables, $u_i, v_i$ and $w_i$ as 
\begin{align*}
&u_i \sim {\rm{Ga}} ( u_i | 1 + c, 1 + \log \{ 1 + \log (1 + \eta _i ) \} ),  \ \ \ \ 
v_i \sim {\rm{Ga}} ( v_i | 1 + u_i , 1 + \log (1 + \eta _i)),\\
&w_i \sim {\rm{Ga}} ( w_i | 1 + v_i , 1 + 1 / \eta _i ). 
\end{align*}
where $\eta _i = {\etat _i}^{1 - z_i + z_i \ga }$. 
Then, 
generate $\etat _i = 1 / {e_i}^{1 / \ga }$ with $e_i \sim {\rm{Ga}} ( v_i + \alt / \ga ^2 , w_i + ( \alt / \ga ^2 ) \exp ( \x_i^{\top} \bbet \ga ) {t_i}^{\ga } )$ when $z_i=1$, and generate $\tilde{\eta}_i$ as $\exp[ \exp \{ b_i/(1-b_i) \} - 1] - 1 $ with $b_i\sim {\rm Beta}(1,c)$ when $z_i=0$.

\item 
(Sampling of $s$) \ Generate $s$ from ${\rm{Beta}}( a_s + \sum_{i = 1}^{n} z_i , b_s +n -\sum_{i = 1}^{n} z_i)$.

\item
(Sampling of $\alpha$) \ The full conditional distribution of $\alpha$ is proportional to 
\begin{align*}
&\alpha^{c_{\al } - 1} e^{- a_{\al }\alpha - b_{\al } /\alpha}  
{\alpha^{n\alpha} \over \Ga ( \alpha )^n} 
\exp \Big[ \alpha\sum_{i = 1}^{n} \Big\{ \x_i^{\top} \bbe + \log {{t_i}^{\ga } \over {\etat _i}^{z_i \ga }} - \exp ( \x_i^{\top} \bbe ) {{t_i}^{\ga } \over {\etat _i}^{z_i \ga }} \Big\} \Big]
\end{align*}
with $\alpha=\alt /\gamma^2$. 
We generate $\alt$ by using the independent Metropolis-Hastings algorithm of \cite{Miller2019}, where the details are given in the Supplementary Material.

\item
(Sampling of $\bbe$) \ The full conditional distribution of $\bbe$ is proportional to 
\begin{align*}
&\exp \Big[
- \frac12 ( \bbe-\b_{\bbe })^{\top } \A _{\bbe } ( \bbe - \b _{\bbe } )  
+
\alpha\sum_{i = 1}^{n} \Big\{ \x_i^{\top} \bbe + \log {{t_i}^{\ga } \over {\etat _i}^{z_i \ga }} - \exp ( \x_i^{\top} \bbe ) {{t_i}^{\ga } \over {\etat _i}^{z_i \ga }} \Big\}
\Big]
\end{align*}
with $\bbe =\gamma\bbet$.
We generate $\bbet$ by using the independent Metropolis-Hastings algorithm of \cite{Miller2019}, where the details are given in the Supplementary Material.

\item
(Sampling of $\ga $) \ 
The full conditional distribution of $\gat$ is proportional to $p_{\gamma}(\gat / (1 - \gat ))/(1 - \gat )^2$, where $p_{\gamma}( \ga)$ is proportional to 
\begin{align}
&{\ga ^p \over \ga ^2} \Big( {\alt \over \ga ^2} \Big) ^{c_{\al } - 1} e^{- a_{\al } ( \alt / \ga ^2 ) - b_{\al } / ( \alt / \ga ^2 )} \exp \Big\{ - {( \bbet \ga - \b _{\bbe } )^{\top } \A _{\bbe } ( \bbet \ga - \b _{\bbe } ) \over 2} \Big\} \ga ^{c_{\ga } - 1} e^{- a_{\ga } \ga - b_{\ga } / \ga } \non \\
&\times \Big( \prod_{\substack{1 \le i \le n \\ z_i = 1}} {\ga \over \tilde{\eta } _i} {{\tilde{\eta } _i}^{\ga } \over 1 + {\tilde{\eta } _i}^{\ga }} {1 \over 1 + \log (1 + {\tilde{\eta } _i}^{\ga } )} {1 \over [1 + \log \{ 1 + \log (1 + {\tilde{\eta } _i}^{\ga } ) \} ]^{1 + c}} \Big) \non \\
&\times \ga ^n \Big\{ {( \alt / \ga ^2 )^{\alt / \ga ^2} \over \Ga ( \alt / \ga ^2 )} \Big\} ^n \exp \Big[ {\alt \over \ga ^2} \sum_{i = 1}^{n} \Big\{ \x_i^{\top} \bbet \ga + \log {{t_i}^{\ga } \over {\tilde{\eta } _i}^{z_i \ga }} - \exp ( \x_i^{\top} \bbet \ga ) {{t_i}^{\ga } \over {\tilde{\eta } _i}^{z_i \ga }} \Big\} \Big].
\end{align}
We use an independent Metropolis-Hastings algorithm to generate $\gat$ and obtain $\ga = \gat / (1 - \gat )$. 
See the Supplementary Material for details. 
\end{itemize}

\section{Theoretical properties}
\label{sec:3}
We here discuss theoretical robustness of the proposed model. 
For simplicity, we assume that $\de _i = 1$ for all $i = 1, \dots , n$; otherwise, we would assume that $\de _i = 1$ whenever $i$ corresponds to an outlier. 
Suppose that $t_1 , \dots , t_n$ are survival times following the GG distribution, namely, $t_i \sim p_{\mathrm{GG}} (\alpha , \gamma , \exp ( \x_i^{\top} \bbe ) / \la _i )$. 
Further, for each $i = 1, \dots , n$, the local variable $\la _i$ is distributed as the mixture DLH distribution (\ref{eq:local_double_log_right}) 
with the first component for non-outliers, $\de _1$, replaced for mathematical tractability by an approximate bounded density $\pi _0$ on $(0, \infty )$, such as ${\rm{IG}} (100, 100)$, which has lighter tails than the super heavy-tailed second component $G_{\mathrm{DLH}}$ given by (\ref{eq:DLH_2025}). 
The prior distribution for the global parameters is decomposed as $\pi ( \al , \bbe , \ga ) = \pi ( \al , \ga ) \pi ( \bbe | \al , \ga )$ and the induced priors are given by $\pit ( \al , \bbet , \ga ) = \pi ( \al , \ga ) \pit ( \bbet | \al , \ga ) = \pi ( \al , \ga ) \ga ^p \pi ( \ga \bbet | \al , \ga )$. 

Let $\Kc , \Lc \subset \{ 1, \dots , n \} $ satisfy $\Kc \cup \Lc = \{ 1, \dots , n \} $, $\Kc \cap \Lc = \emptyset $, and $\Lc \neq \emptyset $, representing the sets of indices of non-outliers and outliers, respectively. 
Following a settings of exploring fully Bayesian robustness \citep[e.g.][]{Gagnon2020}, suppose that while $t_i$ is fixed for $i \in \Kc $, there exist $a_i \in \mathbb{R}$ and $b_i > 0$ satisfying $\log t_i = a_i + b_i \om $ for $i \in \Lc $, where $\om \to \infty $. 
We say that the posterior is robust if 
\begin{align}
\lim_{\om \to \infty } p( \al , \bbet , \ga \mid \mathcal{T} ) 
= p( \al , \bbet , \ga \mid \{ t_i | i \in \Kc \} ) \text{.} \non 
\end{align}
This definition has been adopted in the context of Bayesian posterior robustness \citep[e.g.][]{d2013, d2015,Gagnon2020}. 
This clearly means that the effects of outliers are automatically removed as $\om \to \infty $.

\begin{thm}
\label{thm:robustness_main} 
Suppose that there exist $a_{\bbet } > 0$, $b_{\bbet } > 0$, and $M > 0$ such that for all $( \al , \bbet , \ga ) \in (0, \infty ) \times \mathbb{R} ^p \times (0, \infty )$, we have 
\begin{align}
\pit ( \bbet | \al , \ga ) &\le M \prod_{k = 1}^{p} {| \bet _k |^{a_{\bbet } - 1} \over (1 + | \bet _k |)^{a_{\bbet } + b_{\bbet }}} \text{.} \non 
\end{align}
Suppose that there exists $\ka > 0$ such that 
\begin{align}
E \Big[ {1 \over \ga ^{\ka }} + \ga ^{\ka } + {1 \over \{ ( \min \{ \al , \al ^{1 / 2} \} ) \ga \} ^{\ka }} \Big] < \infty \text{.} \non 
\end{align}
Then the posterior is robust. 
\end{thm}

The proof of the theorem is given in Section \ref{sec:proof} of the Supplementary Material. %
In effect, the first assumption means that Theorem \ref{thm:robustness_main} can be applied only when we use an independent prior for $( \al , \bbet , \ga )$. 
However, the two assumptions are related only to the tails of the priors. 
Moreover, since the constants $a_{\bbet }$, $b_{\bbet }$, and $\ka $ can be chosen to be arbitrarily small (e.g., we might let $\ka = d_0 / 2$ if the marginal density of $\ga $ is proportional to $\ga ^{- 1 - d_0}$ as $\ga \to \infty $ for $d_0 > 0$), virtually no restrictions are imposed on tail behavior.

Theorem~\ref{thm:robustness_main} indicates that the posterior distribution of  proposed model is entirely robust, that is, not only points estimates (such as the posterior means) but also uncertainty quantification (such as the credible intervals) is valid under existence of outliers. 
A key ingredient of the posterior robustness is the heavy-tailed property of the proposed DLH distribution (with double log-adjustment) for the local parameter. 
In the Supplementary Material, we show that the double log-adjustment is essential as the slightly lighter-tailed distribution with single log-adjustment does not hold posterior robustness. 
Moreover, we also show that the log-$t$ distribution \citep[e.g.][]{VallejosSteel2015} does not hold posterior robustness, either, which is also given in the Supplementary Material.

\section{Simulation study}
\label{sec:sim}

We evaluate the performance of the proposed methods through simulation studies. 
We set $n=200$ (number of observations) and $p=3$ (dimension of regression coefficient) throughout this study.
For $i=1,\ldots,n$, two covariates, $x_{i1}$ and $x_{i2}$ are independently generated as $x_{i1}\sim U(0,2)$ and $x_{i2}\sim U(-2,2)$. 
We generated (uncensored and uncontaminated) survival observations from the following two scenarios: 
\begin{align*}
&({\rm GA}) \ \ \ t_i^{\ast}\sim {\rm Ga}(\alpha, \alpha/\theta_i), \ \ \ \ \theta_i=\exp(\beta_0+\beta_1 x_{i1} +\beta_2 x_{i2}), \\
&({\rm GG}) \ \ \ t_i^{\ast}\sim {\rm GG}(\alpha, \gamma, \theta_i), \ \ \ \ \theta_i=\exp(\beta_0+\beta_1 x_{i1} +\beta_2 x_{i2}),
\end{align*}
where the true parameter values are set to $(\beta_0, \beta_1, \beta_2)=(0.5, 2, -0.5)$, $\alpha=10$ in scenario-(GA), and $(\beta_0, \beta_1, \beta_2)=(4, 1, -1)$, $(\alpha,\gamma)=(5,2)$ in scenario-(GG).
For observations with $x_{i2}>0.5$, we generated a binary indicator of outliers as $z_i\sim {\rm Ber}(\omega)$, and set the observed data $t_i$ as $t_i=t_i^{\ast}$ when $z_i=0$ and $t_i=t_i^{\ast}+100$ when $z_i=1$.
The censoring time $C_i$ is generated from $C_i|(z_i=0)\sim U(50, c_{\rm max})$ and $C_i|(z_i=1)=\infty$ (no censoring for outliers), where $c_{\rm max}=\max({t_1^{\ast},\ldots,t_n^{\ast}, 50})$.
Then, the final observation is defined as ${\rm min}(t_i, C_i)$. 
Note that the average censoring rate is around $20\%$ in both scenarios.
In this study, we define the target parameter as the conditional expectation of the genuine observation $E[t^{\ast}|x]$ evaluated at the three points, $(x_1, x_2)=(0.5, -1), (1, 0)$ and $(1.5, 1)$, denoted by reg1, reg2 and reg3, respectively.

For the generated data, we applied the proposed models, RGG (robust generalized gamma), RGA (robust gamma), and RWB (robust Weibull) models.
For comparison, we also adopted their standard versions, GG (generalized gamma), GA (gamma) and WB (Weibull), obtained by setting $\lambda_i=1$ in the robust model.
For the six methods, 2000 posterior samples are generated after discarding the first 2000 samples as burn-in. 
Based on the posterior samples, we computed the posterior means and $95\%$ credible intervals of 
reg1, reg2 and reg3.
To evaluate the performance, we calculated the mean squared errors (MSE) of the posterior means and coverage probabilities (CP) of the credible intervals, based on 500 Monte Carlo replications. 
The results are reported in Tables~\ref{tab:MSE} and \ref{tab:CP}.
In Scenario-(GA) with $\omega=0$, both GG and GA perform quite well as expected, while their robust version, RGG and RGA also exhibit comparable performance, indicating that the proposed robust methods are efficient even under no contamination. 
On the other hand, once outliers are included in the data, non-robust methods are highly affected and fail to provide reasonable results, while the robust methods can still give accurate point estimates. 
Similar results can be found in the interval estimation. 
When outliers are not included, all the methods produce credible intervals with CP around the nominal level. 
However, once outliers exist, CP of the non-robust methods considerable small while those of the robust methods are stable.

\begin{table}[htb!]
\caption{Logarithmic Mean squared errors (MSE) of six methods for three regression values, based on 500 Monte Carlo replications.}
\begin{center}
\begin{tabular}{ccccccccccccccccccc} 
\hline
Scenario & $\omega$ &  &  & RGG & GG & RGA & GA & RWB & WB \\
 \hline
 &  & reg1 &  & 0.07 & 0.07 & 0.06 & 0.06 & 0.07 & 0.08 \\
(GA) & 0\% & reg2 &  & 0.07 & 0.07 & 0.07 & 0.07 & 0.08 & 0.07 \\
 &  & reg3 &  & 0.45 & 0.46 & 0.45 & 0.45 & 0.49 & 0.50 \\
\hline
 &  & reg1 &  & 0.07 & 1.99 & 0.52 & 2.75 & 0.25 & 2.75 \\
(GA) & 5\% & reg2 &  & 0.08 & 3.48 & 0.71 & 4.35 & 0.63 & 4.38 \\
 &  & reg3 &  & 0.50 & 5.05 & 1.26 & 6.12 & 1.89 & 6.17 \\
 \hline
 &  & reg1 &  & 0.09 & 2.99 & 0.23 & 3.40 & 0.22 & 3.40 \\
(GA) & 10\% & reg2 &  & 0.09 & 4.70 & 0.44 & 5.24 & 0.25 & 5.24 \\
 &  & reg3 &  & 0.54 & 6.40 & 1.18 & 7.11 & 0.87 & 7.12 \\
\hline
 &  & reg1 &  & 0.16 & 0.18 & 0.16 & 0.16 & 0.19 & 0.24 \\
(GG) & 0\% & reg2 &  & 0.04 & 0.05 & 0.04 & 0.04 & 0.05 & 0.06 \\
 &  & reg3 &  & 0.06 & 0.06 & 0.06 & 0.06 & 0.07 & 0.09 \\
 \hline
 &  & reg1 &  & 0.14 & 2.48 & 0.53 & 2.09 & 1.05 & 2.06 \\
(GG) & 5\% & reg2 &  & 0.04 & 3.07 & 0.38 & 3.69 & 1.02 & 3.79 \\
 &  & reg3 &  & 0.06 & 3.44 & 0.58 & 4.64 & 1.88 & 4.82 \\
 \hline
 &  & reg1 &  & 0.67 & 3.39 & 0.39 & 2.56 & 1.03 & 2.53 \\
(GG) & 10\% & reg2 &  & 1.19 & 4.24 & 1.50 & 4.62 & 2.08 & 4.66 \\
 &  & reg3 &  & 1.69 & 4.80 & 2.52 & 5.73 & 3.09 & 5.82 \\
 \hline
\end{tabular}
\end{center}
\label{tab:MSE}
\end{table}

\begin{table}[htb!]
\caption{Coverage probabilities (CP) of six methods for three regression values, based on 500 Monte Carlo replications.}
\begin{center}
\begin{tabular}{ccccccccccccccccccc} 
\hline
Scenario & $\omega$ &  &  & RGG & GG & RGA & GA & RWB & WB \\
 \hline
 &  & reg1 &  & 95.0 & 95.0 & 96.2 & 95.6 & 93.0 & 93.0 \\
(GA)& 0\% & reg2 &  & 96.0 & 96.4 & 96.8 & 96.6 & 96.0 & 96.8 \\
 &  & reg3 &  & 93.0 & 92.8 & 93.2 & 93.4 & 91.0 & 90.6 \\
 \hline
 &  & reg1 &  & 94.4 & 20.5 & 92.8 & 22.8 & 88.3 & 25.3 \\
(GA) & 5\% & reg2 &  & 94.1 & 0.2 & 91.4 & 0.2 & 86.7 & 0.2 \\
 &  & reg3 &  & 93.7 & 0.5 & 91.6 & 0.0 & 86.2 & 0.0 \\
 \hline
 &  & reg1 &  & 92.3 & 2.6 & 93.9 & 3.7 & 89.8 & 6.1 \\
(GA) & 10\% & reg2 &  & 92.3 & 0.0 & 93.9 & 0.0 & 90.9 & 0.0 \\
 &  & reg3 &  & 93.7 & 0.0 & 94.7 & 0.0 & 92.7 & 0.0 \\
\hline
 &  & reg1 &  & 93.0 & 92.8 & 94.2 & 94.4 & 92.4 & 91.6 \\
(GG) & 0\% & reg2 &  & 92.6 & 92.2 & 93.4 & 93.8 & 91.2 & 92.6 \\
 &  & reg3 &  & 93.0 & 93.2 & 93.8 & 94.0 & 92.0 & 91.4 \\
 \hline
 &  & reg1 &  & 94.7 & 38.6 & 93.3 & 74.4 & 71.4 & 80.7 \\
(GG) & 5\% & reg2 &  & 95.1 & 0.2 & 93.3 & 0.2 & 71.2 & 0.2 \\
 &  & reg3 &  & 94.7 & 0.2 & 93.0 & 0.2 & 72.1 & 0.2 \\
 \hline
 &  & reg1 &  & 92.9 & 11.9 & 95.6 & 60.5 & 84.9 & 69.0 \\
(GG) & 10\% & reg2 &  & 90.7 & 0.0 & 91.7 & 0.0 & 82.5 & 0.0 \\
 &  & reg3 &  & 92.5 & 0.0 & 94.0 & 0.0 & 82.3 & 0.0 \\
 \hline
\end{tabular}
\end{center}
\label{tab:CP}
\end{table}

\section{Real data example}
\label{sec:analysis}

\subsection{Setting}
We apply the proposed methods to real data analysis. We used hospital stay data for ``Major cardiovascular interventions" available at \cite{locatelli2011robust}, where the date of admission and discharge is observed for 75 individuals. 45 stays were censored because patients were transferred to a different hospital prior to dismissal. 
We are interested in analyzing the relationship between the length of stay and two covariates: age of the patient ($x_1$) and type of admission ($x_2=0$ for regular admissions, $x_2=1$ for emergency admissions). The data is shown in Figure \ref{fig:reg-complete}. From the figure, it is observed that the length of stay of two young patients was exceptionally high, and there are some outliers for patients over 50 years old. Let $t_i$ be the length of stay for $i$th patient, and consider the accelerated failure time model given by \eqref{GG}, where
\begin{equation*}
\theta_i=\exp(\beta_0 + \beta_1 x_{i1} + \beta_2 x_{i2} + \beta_3 x_{i1} x_{i2})/\lambda_i, \quad i=1,\dots,75
\end{equation*}
Note that $x_1 x_2$ is the interaction term of $x_1$ and $x_2$. We apply the proposed RGG and RGA methods as well as the GG, GA and LST methods. The LST method is the log-Student $t$ model proposed by \cite{VallejosSteel2015}\footnote{R code to implement the LST method is available at \url{https://warwick.ac.uk/fac/sci/statistics/staff/academic-research/steel/steel_homepage/software}}. The model is represented as the shape mixture of log-normal distributions given by \eqref{steel}. When the mixing density is the gamma distribution $\mathrm{Ga}(\nu/2,\nu/2)$, the corresponding model is the LST model. For these methods, we generated 20000 posterior samples and only every 5th scan was saved (thinning). 

We calculated the posterior medians of the random effect $\lambda_i$ for the RGG, RGA, and LST methods to quantify the unobserved heterogeneity. Next, we compute posterior quantiles for the values of the regression function for specific explanatory variables as well as the model selection criterion called the deviation information criterion \citep[DIC;][]{spiegelhalter2002bayesian}. 
To see the robustness of the proposed methods, we also compared the estimated regression curves for the data with or without outliers.

\subsection{Results}

We report the posterior medians of the local parameter $\lambda_i$ for three robust methods in Figure \ref{fig:lambda-plot}. For RGG and RGA, the logarithmic values of the posterior median for non-outlying observations stick to 0, while for LST they are observed to be scattered near 0. This means that the role of $\lambda_i$ for outliers is clear in the proposed method. The uncertainty of outliers for the RGG and RGA methods is also presented in Figure \ref{fig:lambda-plot}. In the proposed models, we can quantify the uncertainty of outliers through the marginal posterior probability of $z_i=1$, where $z_i$ is the latent binary variable defined in Section %
\ref{subsec:22}. We show summary statistics of the posterior distribution of the regression function for specific explanatory variables in Table \ref{tab:app}. We here focus on 10 and 80 years old. In terms of posterior median, the two robust proposed methods seem to be comparable, while it is observed that the GG and GAM methods are affected by outliers. The length of 95\% posterior credible intervals of the RGA method is shorter than that of the other four methods. For the LST\footnote{Theoretically, it is expected to be robust only ``partially". 
See Section \ref{subsec:supp-logT} of the Supplementary Material. On the other hand, the proposed model has ``whole" robustness.}, the performance is intermediate between them. The smallest DIC value also suggests that the RGA model is reasonable for the data set. Finally, we show the results of the estimated regression curves for five methods in Figure \ref{fig:reg-complete}. The left panel shows (log-scale) estimated regression curves for regular admission ($x_2=0$) defined by \eqref{expectation-gg} and \eqref{expectation-LN} with $\lambda_i=1$, where the regression part is written as $\hat{\beta}_0 + \hat{\beta}_1 \times \text{Age}$. The right panel is the case of emergency admission ($x_2=1$), where the regression part is written as $(\hat{\beta}_0 + \hat{\beta}_2) + (\hat{\beta}_1 + \hat{\beta}_3) \times \text{Age}$. Note that we used the posterior median for the point estimates of $\beta$. These results show that for $x_2=0$, the slopes of the regression curves for the five methods are the same, but the effect of outliers is evident in the vertical variation. For $x_2=1$, RGG and RGA give similar results, but the other three methods have very different slopes due to the effect of outliers.

\begin{figure}
    \centering
    \includegraphics[width=\linewidth]{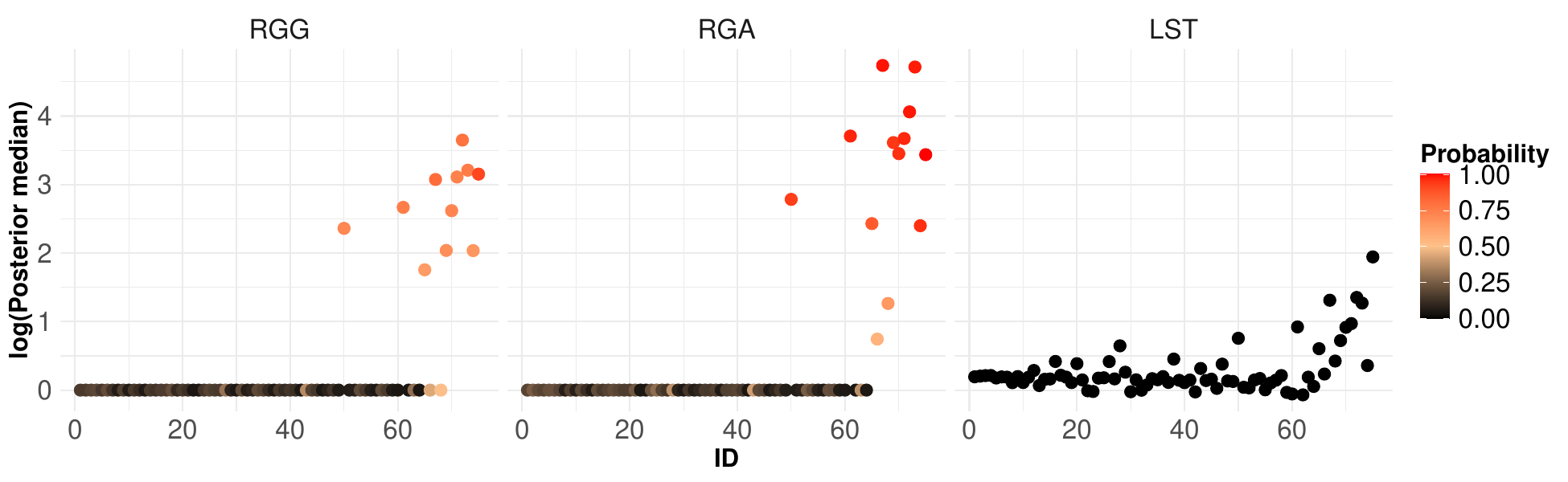}
    \caption{The (log-scale) posterior medians of $\lambda_i$ and marginal posterior probabilities of $z_i=1$ for the RGG (left) and RGA (center), and the (log-scale) posterior medians of $1/\lambda_i$ for the LST (right). As the LST model cannot calculate the marginal posterior probability of outliers, the black points do not indicate that the probability is zero.}
    \label{fig:lambda-plot}
\end{figure}

\begin{table}[htb!]
\caption{Posterior quantiles and deviance information criteria (DIC) of five models. }
\begin{center}
\begin{tabular}{ccccccccccccccccccc} 
\hline
Type & Age & Quantile (\%) &  & RGG & RGA & GG & GA & LST \\
\hline
 &  & 2.5 &  & 8.5 & 8.4 & 24.0 & 29.6 & 9.0 \\
0 & 10 & 50 &  & 11.2 & 10.6 & 41.3 & 45.1 & 17.8 \\
 &  & 97.5 &  & 49.5 & 15.7 & 81.3 & 74.8 & 52.0 \\
 \hline
 &  & 2.5 &  & 7.3 & 7.2 & 55.1 & 50.7 & 11.9 \\
1 & 10 & 50 &  & 29.3 & 22.5 & 214.8 & 128.0 & 181.3 \\
 &  & 97.5 &  & 438.0 & 124.7 & 1252.9 & 384.3 & 1122.7 \\
 \hline
 &  & 2.5 &  & 16.5 & 16.3 & 43.2 & 43.1 & 17.5 \\
0 & 80 & 50 &  & 24.4 & 21.8 & 86.9 & 83.9 & 44.5 \\
 &  & 97.5 &  & 134.1 & 53.1 & 227.8 & 222.6 & 152.1 \\
 \hline
 &  & 2.5 &  & 27.4 & 27.1 & 40.6 & 35.9 & 28.3 \\
1 & 80 & 50 &  & 36.9 & 34.7 & 76.8 & 66.3 & 47.5 \\
 &  & 97.5 &  & 103.8 & 52.6 & 203.5 & 169.2 & 128.6 \\
 \hline
\multicolumn{3}{c}{DIC} &  & -478.4 & -761.5 & 736.4 & 729.6 & 290.0 \\
 \hline
\end{tabular}
\end{center}
\label{tab:app}
\end{table}

\begin{figure}
    \centering
    \includegraphics[width=\linewidth]{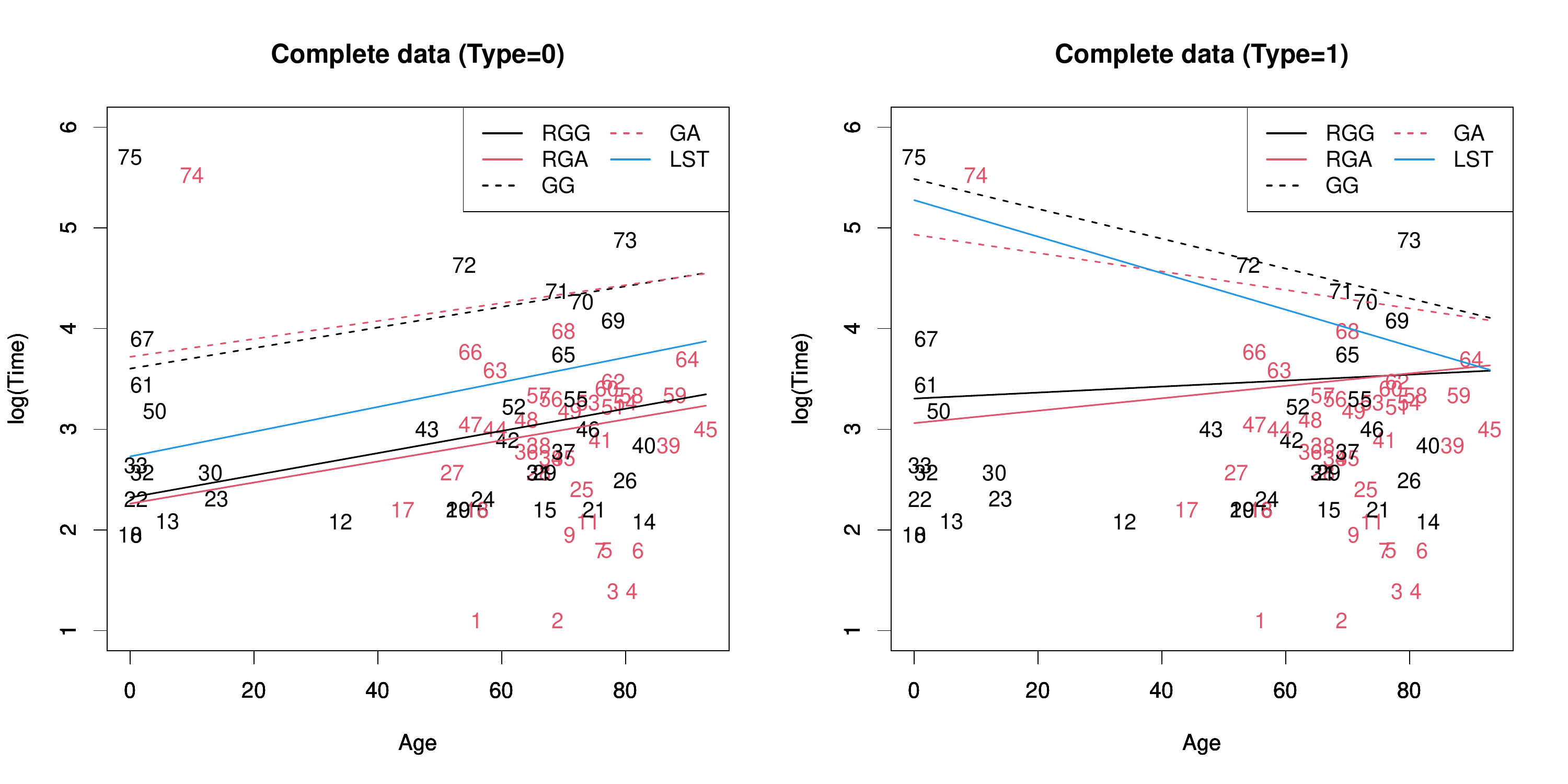}
    \caption{Estimated (log-scale) regression curves for regular ($x_2=0$) and emergency ($x_2=1$) admissions. Observations colored by red are patients with emergency admission. }
    \label{fig:reg-complete}
\end{figure}

\subsection{Robustness}

We confirm the robustness of the proposed method using the data after we have detected the outliers. Since Table \ref{tab:app} shows that RGA is the model with the smallest DIC among the five methods, we will use the outlier detection results using RGA. Since RGA enables outlier detection using the posterior probability of $z_i$, observations for which the posterior probability is greater than 0.5 are detected as outliers and excluded from the data. Figure \ref{fig:reg-without-outliers} shows the results of applying the five methods to the data excluding the 13 outliers actually detected by the RGA. When outliers are detected by RGA, all methods give almost the same estimation results (Figure \ref{fig:reg-without-outliers}), and RGA seems to be able to identify outliers well. 
Furthermore, this result shows that RGG and RGA can be stably estimated with or without outliers. In particular, we also observed that RGA performs slightly better than RGG.

\begin{figure}
    \centering
    \includegraphics[width=\linewidth]{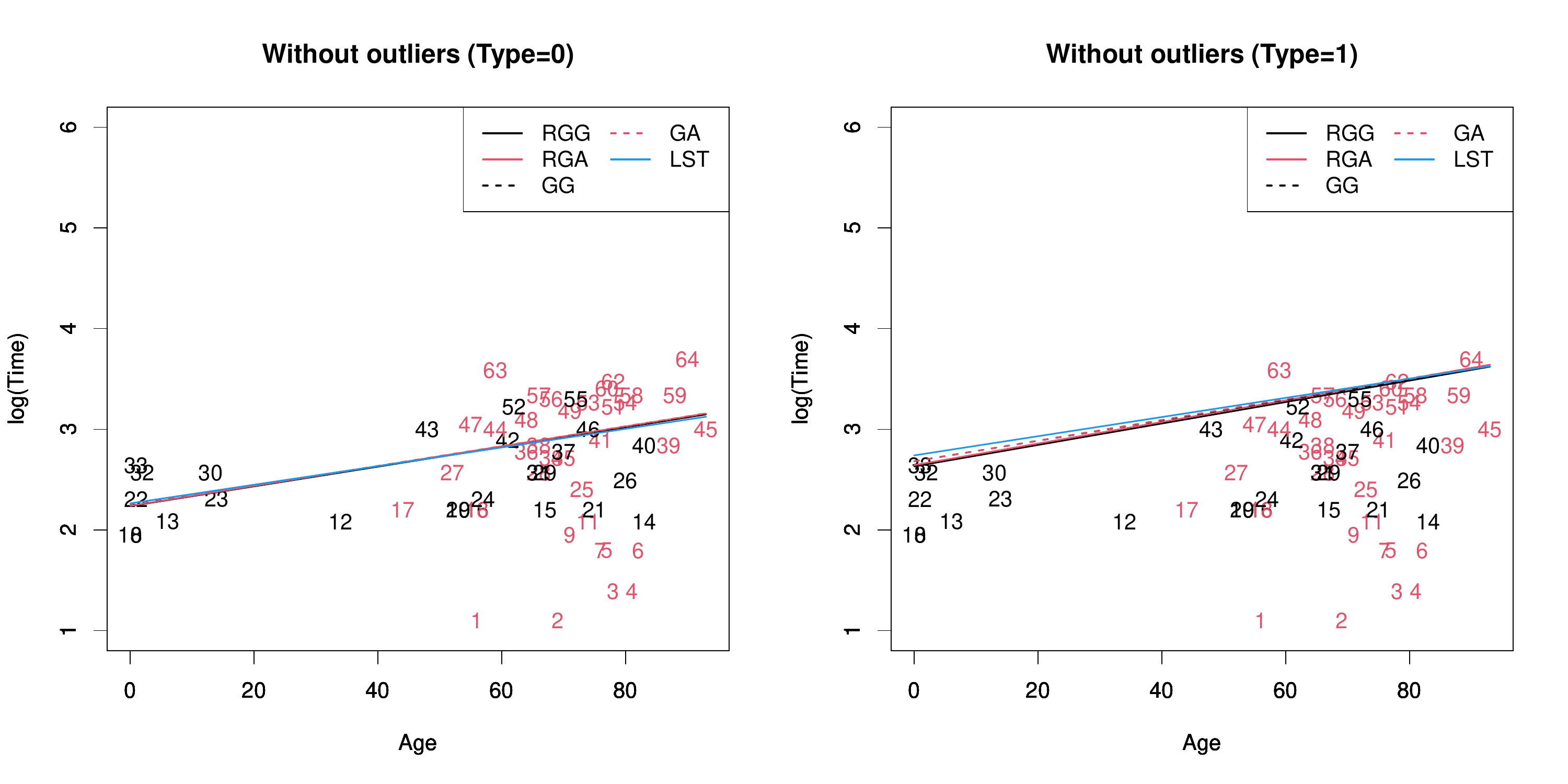}
    \caption{Estimated (log-scale) regression curves for regular ($x_2=0$) and emergency ($x_2 =1$) admissions when outliers are removed by using the RGA method. Observations colored by red are patients with emergency admission. }
    \label{fig:reg-without-outliers}
\end{figure}

\section{Concluding remarks}

We proposed a robust accelerated failure time model based on a scale mixture of generalized gamma distributions, which has three notable features. Considering a heavy-tailed mixing density, we constructed the super heavy-tailed distribution which makes robust inference on the parameters. The proposed model provides an interpretable structure for conditional expectation of survival time, and we can quantify the uncertainty of outliers using the latent binary variable. We rigorously showed that the proposed method has the desirable robustness property, called posterior robustness, under some mild conditions.

The proposed model can be extended to a multi-level model. 
For example, if we have survival times in multiple hospitals, hospital-specific random effects can easily be incorporated into the proposed model. Furthermore, the proposed model is also expected to have applications in finance, such as bankruptcy prediction \citep[e.g.][]{pierri2017bankruptcy}, in addition to the field of biostatistics.

Finally, the proposed method with generalized gamma distribution can be extended to a more general class of density function, described as  
\begin{equation*}
p(t_i|\theta_i, \phi)=f(t_i,\phi)\theta_i^{g(\phi)}\exp\left\{-\theta_ih(t_i,\phi)\right\}, \ \ \ \ i=1,\ldots,n,
\end{equation*}
where functions $f(t_i,\phi)$, $g(\phi)$ and $h(t_i,\phi)$ are chosen such that it is a proper probability density. 
We assume that $f(t_i,\phi)$ and $h(t_i,\phi)$ are polynomial functions of $t_i$.
This class of probability function is a subclass of the exponential model considered in \cite{aktekin2020family}, and it includes many well-known family of models such as the Poisson, exponential, Weibull, gamma, and generalized gamma distributions. 
By setting, $\bphi=(\alpha,\gamma)$, $f(t_i,\bphi)=\gamma t_i^{\alpha\gamma-1}\alpha^\alpha /\Gamma(\alpha)$, $g(\bphi)=\alpha$ and $h(t_i,\bphi)=\alpha t_i^\gamma$, the above model reduces to the generalized gamma model considered in this paper. 
While the proposed posterior computation algorithm can be applied for $\theta_i$, the detailed investigation of the theoretical robustness will be a valuable future work.

\section*{Acknowledgments}
This work is partially supported by Japan Society for Promotion of Science (KAKENHI) grant numbers 25K21163, 22K20132, 19K11852, 21K13835, and 21H00699.

\vspace{1cm}
\bibliographystyle{chicago}
\bibliography{ref}

\begin{thebibliography}{}

\bibitem[\protect\citeauthoryear{Aktekin, Polson, and Soyer}{Aktekin
  et~al.}{2020}]{aktekin2020family}
Aktekin, T., N.~G. Polson, and R.~Soyer (2020).
\newblock A family of multivariate non-gaussian time series models.
\newblock {\em Journal of Time Series Analysis\/}~{\em 41\/}(5), 691--721.

\bibitem[\protect\citeauthoryear{Cox, Chu, Schneider, and Munoz}{Cox
  et~al.}{2007}]{cox2007parametric}
Cox, C., H.~Chu, M.~F. Schneider, and A.~Munoz (2007).
\newblock Parametric survival analysis and taxonomy of hazard functions for the
  generalized gamma distribution.
\newblock {\em Statistics in Medicine\/}~{\em 26\/}(23), 4352--4374.

\bibitem[\protect\citeauthoryear{Cox}{Cox}{1972}]{cox1972regression}
Cox, D.~R. (1972).
\newblock Regression models and life-tables.
\newblock {\em Journal of the Royal Statistical Society: Series B
  (Methodological)\/}~{\em 34\/}(2), 187--202.

\bibitem[\protect\citeauthoryear{Desgagn\'{e}}{Desgagn\'{e}}{2013}]{d2013}
Desgagn\'{e}, A. (2013).
\newblock Full robustness in bayesian modelling of a scale parameter.
\newblock {\em Bayesian Analysis\/}~{\em 8}, 187--220.

\bibitem[\protect\citeauthoryear{Desgagn\'{e}}{Desgagn\'{e}}{2015}]{d2015}
Desgagn\'{e}, A. (2015).
\newblock Robustness to outliers in location-scale parameter model using
  log-regularly varying distributions.
\newblock {\em Annals of Statistics\/}~{\em 43}, 1568--1595.

\bibitem[\protect\citeauthoryear{Gagnon, Desgagn{\'e}, and B{\'e}dard}{Gagnon
  et~al.}{2020}]{Gagnon2020}
Gagnon, P., A.~Desgagn{\'e}, and M.~B{\'e}dard (2020).
\newblock A new {B}ayesian approach to robustness against outliers in linear
  regression.
\newblock {\em Bayesian Analysis\/}~{\em 15\/}(2), 389--414.

\bibitem[\protect\citeauthoryear{Gagnon and Hayashi}{Gagnon and
  Hayashi}{2023}]{gagnon2023theoretical}
Gagnon, P. and Y.~Hayashi (2023).
\newblock Theoretical properties of bayesian student-t linear regression.
\newblock {\em Statistics \& Probability Letters\/}~{\em 193}, 109693.

\bibitem[\protect\citeauthoryear{Hamura, Irie, and Sugasawa}{Hamura
  et~al.}{2020}]{hamura2020shrinkage}
Hamura, Y., K.~Irie, and S.~Sugasawa (2020).
\newblock Shrinkage with robustness: Log-adjusted priors for sparse signals.
\newblock {\em arXiv preprint arXiv:2001.08465\/}.

\bibitem[\protect\citeauthoryear{Hamura, Irie, and Sugasawa}{Hamura
  et~al.}{2022}]{Hamura2022log}
Hamura, Y., K.~Irie, and S.~Sugasawa (2022).
\newblock Log-regularly varying scale mixture of normals for robust regression.
\newblock {\em Computational Statistics \& Data Analysis\/}~{\em 173}, 107517.

\bibitem[\protect\citeauthoryear{Hamura, Irie, and Sugasawa}{Hamura
  et~al.}{2024a}]{hamura2024posterior}
Hamura, Y., K.~Irie, and S.~Sugasawa (2024a).
\newblock Posterior robustness with milder conditions: Contamination models
  revisited.
\newblock {\em Statistics \& Probability Letters\/}~{\em 210}, 110130.

\bibitem[\protect\citeauthoryear{Hamura, Irie, and Sugasawa}{Hamura
  et~al.}{2024b}]{hamura2024robust}
Hamura, Y., K.~Irie, and S.~Sugasawa (2024b).
\newblock Robust {B}ayesian modeling of counts with zero inflation and
  outliers: Theoretical robustness and efficient computation.
\newblock {\em Journal of the American Statistical Association\/}, 1--19.

\bibitem[\protect\citeauthoryear{Hamura, Onizuka, Hashimoto, and
  Sugasawa}{Hamura et~al.}{2024}]{hohs2022}
Hamura, Y., T.~Onizuka, S.~Hashimoto, and S.~Sugasawa (2024).
\newblock Sparse {B}ayesian inference on gamma-distributed observations using
  shape-scale inverse-gamma mixtures.
\newblock {\em Bayesian Analysis\/}~{\em 19\/}(1), 77--97.

\bibitem[\protect\citeauthoryear{Klein, Van~Houwelingen, Ibrahim, and
  Scheike}{Klein et~al.}{2014}]{klein2014handbook}
Klein, J.~P., H.~C. Van~Houwelingen, J.~G. Ibrahim, and T.~H. Scheike (2014).
\newblock {\em Handbook of Survival Analysis}.
\newblock CRC Press Boca Raton, FL:.

\bibitem[\protect\citeauthoryear{Lachos, Dey, Cancho, and Louzada}{Lachos
  et~al.}{2017}]{Lachos2017}
Lachos, V.~H., D.~K. Dey, V.~G. Cancho, and F.~Louzada (2017).
\newblock Scale mixtures log-birnbaum--saunders regression models with censored
  data: a bayesian approach.
\newblock {\em Journal of Statistical Computation and Simulation\/}~{\em
  87\/}(10), 2002--2022.

\bibitem[\protect\citeauthoryear{Locatelli, Marazzi, and Yohai}{Locatelli
  et~al.}{2011}]{locatelli2011robust}
Locatelli, I., A.~Marazzi, and V.~J. Yohai (2011).
\newblock Robust accelerated failure time regression.
\newblock {\em Computational Statistics \& Data Analysis\/}~{\em 55\/}(1),
  874--887.

\bibitem[\protect\citeauthoryear{Luengo, Martino, Bugallo, Elvira, and
  S{\"a}rkk{\"a}}{Luengo et~al.}{2020}]{luengo2020survey}
Luengo, D., L.~Martino, M.~Bugallo, V.~Elvira, and S.~S{\"a}rkk{\"a} (2020).
\newblock A survey of monte carlo methods for parameter estimation.
\newblock {\em EURASIP Journal on Advances in Signal Processing\/}~{\em 2020},
  1--62.

\bibitem[\protect\citeauthoryear{Marshall and Olkin}{Marshall and
  Olkin}{2007}]{Marshall2007}
Marshall, A.~W. and I.~Olkin (2007).
\newblock {\em Life distributions}, Volume~13.
\newblock Springer.

\bibitem[\protect\citeauthoryear{Miller}{Miller}{2019}]{Miller2019}
Miller, J.~W. (2019).
\newblock Fast and accurate approximation of the full conditional for gamma
  shape parameters.
\newblock {\em Journal of Computational and Graphical Statistics\/}~{\em
  28\/}(2), 476--480.

\bibitem[\protect\citeauthoryear{Omori and Johnson}{Omori and
  Johnson}{1993}]{Omori1993}
Omori, Y. and R.~A. Johnson (1993).
\newblock The influence of random effects on the unconditional hazard rate and
  survival functions.
\newblock {\em Biometrika\/}~{\em 80\/}(4), 910--914.

\bibitem[\protect\citeauthoryear{Pierri and Caroni}{Pierri and
  Caroni}{2017}]{pierri2017bankruptcy}
Pierri, F. and C.~Caroni (2017).
\newblock Bankruptcy prediction by survival models based on current and lagged
  values of time-varying financial data.
\newblock {\em Communications in Statistics: Case Studies, Data Analysis and
  Applications\/}~{\em 3\/}(3-4), 62--70.

\bibitem[\protect\citeauthoryear{Shukla, Ranjan, and Upadhyay}{Shukla
  et~al.}{2023}]{shukla2023bayes}
Shukla, A., R.~Ranjan, and S.~K. Upadhyay (2023).
\newblock Bayes analysis of the generalized gamma aft models for left truncated
  and right censored data.
\newblock {\em Journal of Statistical Computation and Simulation\/}~{\em
  93\/}(12), 2026--2051.

\bibitem[\protect\citeauthoryear{Spiegelhalter, Best, Carlin, and Van
  Der~Linde}{Spiegelhalter et~al.}{2002}]{spiegelhalter2002bayesian}
Spiegelhalter, D.~J., N.~G. Best, B.~P. Carlin, and A.~Van Der~Linde (2002).
\newblock Bayesian measures of model complexity and fit.
\newblock {\em Journal of the Royal Satistical Society: Series B (Statistical
  Methodology)\/}~{\em 64\/}(4), 583--639.

\bibitem[\protect\citeauthoryear{Stacy}{Stacy}{1962}]{stacy1962generalization}
Stacy, E.~W. (1962).
\newblock A generalization of the gamma distribution.
\newblock {\em The Annals of Mathematical Statistics\/}, 1187--1192.

\bibitem[\protect\citeauthoryear{Vallejos and Steel}{Vallejos and
  Steel}{2015}]{VallejosSteel2015}
Vallejos, C.~A. and M.~F. Steel (2015).
\newblock Objective bayesian survival analysis using shape mixtures of
  log-normal distributions.
\newblock {\em Journal of the American Statistical Association\/}~{\em
  110\/}(510), 697--710.

\bibitem[\protect\citeauthoryear{Vallejos and Steel}{Vallejos and
  Steel}{2017}]{VallejosSteel2017}
Vallejos, C.~A. and M.~F. Steel (2017).
\newblock Incorporating unobserved heterogeneity in weibull survival models: A
  bayesian approach.
\newblock {\em Econometrics and Statistics\/}~{\em 3}, 73--88.

\end{thebibliography}

\newpage

\newpage
\setcounter{page}{1}
\setcounter{equation}{0}
\renewcommand{\theequation}{S\arabic{equation}}
\setcounter{section}{0}
\renewcommand{\thelem}{S\arabic{lem}}
\setcounter{thm}{0}
\renewcommand{\thethm}{S\arabic{thm}}
\setcounter{prp}{0}
\renewcommand{\theprp}{S\arabic{prp}}
\setcounter{table}{0}
\renewcommand{\thesection}{S\arabic{section}}
\setcounter{table}{0}
\renewcommand{\thetable}{S\arabic{table}}
\setcounter{figure}{0}
\renewcommand{\thefigure}{S\arabic{figure}}

\begin{center}
{\LARGE\bf Supplementary Materials for ``Robust Bayesian Inference for Censored Survival Models"}
\end{center}

\vspace{0.5cm}
\begin{center}
{\large 
Yasuyuki Hamura$^{1}$, Takahiro Onizuka$^{2}$, Shintaro Hashimoto$^{3}$ and Shonosuke Sugasawa$^{4}$
}
\end{center}

\medskip
\noindent
$^{1}$Graduate School of Economics, Kyoto University\\
$^{2}$Graduate School of Social Sciences, Chiba University\\
$^{3}$Department of Mathematics, Hiroshima University\\
$^{4}$Faculty of Economics, Keio University

\vspace{1cm}
This Supplementary Material provides additional information about the posterior computation algorithm and the theoretical result of the main text. 
\section{Independent Metropolis-Hastings algorithms for updating global parameters}
\label{sec:MH} 
The global parameters $\alt $, $\bbet $, and $\ga $ are updated using the independent Metropolis-Hastings algorithm. 
In order to construct proposal distributions for $\alt $ and $\bbet $, we use \cite{Miller2019}'s method. 
For $\ga $, we use the piecewise linear approximation \citep[e.g.][Section 3.1.4]{luengo2020survey}.

In each scan of the MCMC algorithm, the full conditional density of $\alt $, namely, $p( \alt | \bbet , \ga , s , Z , E , \mathcal{T} )$, is approximated by a gamma density ${\rm{Ga}} ( \alt | A_{\alt } , B_{\alt } )$, which is used as our proposal distribution in the Metropolis-Hastings algorithm. 
Here, $A_{\alt } , B_{\alt } > 0$ are obtained using \cite{Miller2019}'s iterative method. 
Specifically, given a current value of the tuning parameters $( A_{\alt }^{\rm{old}} , B_{\alt }^{\rm{old}} )$, we compute a new value $( A_{\alt }^{\rm{new}} , B_{\alt }^{\rm{new}} )$ as follows: 
\begin{itemize}
\item
let $C_{\alt }^{\rm{old}} = A_{\alt }^{\rm{old}} / B_{\alt }^{\rm{old}}$ denote the mean of the current proposal distribution; 
\item
solve for $A_{\alt }$ and $B_{\alt }$ the equations 
\begin{align}
&\Big( {\pd \over \pd \alt } \log {\rm{Ga}} ( \alt | A_{\alt } , B_{\alt } ) \Big) \Big| _{\alt = C_{\alt }^{\rm{old}}} = \Big( {\pd \over \pd \alt } \log p( \alt | \bbet , \ga , s , Z , E , \mathcal{T} ) \Big) \Big| _{\alt = C_{\alt }^{\rm{old}}} \text{,} \non \\
&\Big\{ {\pd ^2 \over ( \pd \alt )^2} \log {\rm{Ga}} ( \alt | A_{\alt } , B_{\alt } ) \Big\} \Big| _{\alt = C_{\alt }^{\rm{old}}} = \Big\{ {\pd ^2 \over ( \pd \alt )^2} \log p( \alt | \bbet , \ga , s , Z , E , \mathcal{T} ) \Big\} \Big| _{\alt = C_{\alt }^{\rm{old}}} \text{;} \non 
\end{align}
\item
set $( A_{\alt }^{\rm{new}} , B_{\alt }^{\rm{new}} )$ to the solution. 
\end{itemize}

Similarly, the full conditional density of $\bbet $, namely, $p( \bbet | \alt , \ga , s , Z , E , \mathcal{T} )$, is approximated by a normal proposal density ${\rm{N}}_p ( \bbet | \bmu _{\bbet } , \bPsi _{\bbet } )$, where $\bmu _{\bbet } \in \mathbb{R} ^p$ and $\bPsi _{\bbet } > \O ^{(p)}$ are obtained using \cite{Miller2019}'s idea. 
Specifically, given a current value of the tuning parameters $( \bmu _{\bbet }^{\rm{old}} , \bPsi _{\bbet }^{\rm{old}} )$, we compute a new value $( \bmu _{\bbet }^{\rm{new}} , \bPsi _{\bbet }^{\rm{new}} )$ as follows: 
\begin{itemize}
\item
solve for $\bmu _{\bbet }$ and $\bPsi _{\bbet }$ the equations 
\begin{align}
&\Big( {\pd \over \pd \bbet } \log {\rm{N}}_p ( \bbet | \bmu _{\bbet } , \bPsi _{\bbet } ) \Big) \Big| _{\bbet = \bmu _{\bbet }^{\rm{old}}} = \Big( {\pd \over \pd \bbet } \log p( \bbet | \alt , \ga , s , Z , E , \mathcal{T} ) \Big) \Big| _{\bbet = \bmu _{\bbet }^{\rm{old}}} \text{,} \non \\
&\Big\{ {\pd ^2 \over ( \pd \bbet ) ( \pd \bbet )^{\top }} \log {\rm{N}}_p ( \bbet | \bmu _{\bbet } , \bPsi _{\bbet } ) \Big\} \Big| _{\bbet = C\bmu _{\bbet }^{\rm{old}}} = \Big\{ {\pd ^2 \over ( \pd \bbet ) ( \pd \bbet )^{\top }} \log p( \bbet | \alt , \ga , s , Z , E , \mathcal{T} ) \Big\} \Big| _{\bbet = \bmu _{\bbet }^{\rm{old}}} \text{;} \non 
\end{align}
\item
set $( \bmu _{\bbet }^{\rm{new}} , \bPsi _{\bbet }^{\rm{new}} )$ to the solution. 
\end{itemize}

Finally, we make the change of variables $\gat = \ga / (1 + \ga ) \in (0, 1)$ and use the independent Metropolis-Hastings algorithm to update $\gat $ instead of $\ga $. 
The full conditional density of $\gat $, namely, $p( \gat | \alt , \bbet , s , Z , E , \mathcal{T} )$, is approximated by the normalized version of a piecewise linear function $J( \gat ; \alt , \bbet , s , Z , E , \mathcal{T} )$, $\gat \in (0, 1)$, which satisfies the following properties for some $G \in \mathbb{N}$: 
\begin{itemize}
\item
It is linear on $[(g - 1) / G, g / G] \subset (0, 1)$, $g \in \mathbb{N}$; 
\item
$J(g / G ; \alt , \bbet , s , Z , E , \mathcal{T} ) = p(g / G | \alt , \bbet , s , Z , E , \mathcal{T} )$ for all $g \in \mathbb{N} \cap (0, G)$. 
\end{itemize}
We remark that we can sample directly from such a proposal distribution.

\section{Theoretical results}
\label{sec:proof}
Here, we derive conditions for posterior robustness under a slightly weaker assumption on the local prior $\pi ( \la _i )$ than that of the main text.

\subsection{The main theorem}
\label{subsec:supp-theorem} 
Suppose that $t_1 , \dots , t_n$ follow a generalized gamma distribution with density 
\begin{align}
&{\al ^{\al } \{ \exp ( \x_i^{\top} \bbe ) / \la _i \} ^{\al } \ga \over \Ga ( \al )} {t_i}^{\ga \al - 1} 
\exp\Big\{- \al {t_i}^{\ga } \exp ( \x_i^{\top} \bbe ) / \la _i  \Big\},
\end{align}
for $i=1,\ldots,n$. 
Let $\pi ( \la _i )$ be a density of the local parameter $\lambda_i$,  and  $\pi ( \al , \bbe , \ga )$ be a joint prior of $(\alpha, \bbe, \gamma)$.
In what follows, we write $ \pi ( \al , \bbe , \ga ) = \pi ( \al , \ga ) \pi ( \bbe | \al , \ga )$.

Let $\Kc , \Lc \subset \{ 1, \dots , n \} $ satisfy $\Kc \cup \Lc = \{ 1, \dots , n \} $, $\Kc \cap \Lc = \emptyset $, and $\Lc \neq \emptyset $ and suppose that while $t_i$ are fixed for $i \in \Kc $, there exist $a_i \in \mathbb{R}$ and $b_i > 0$ satisfying $\log t_i = a_i + b_i \om $ for $i \in \Lc $, where $\om \to \infty $.

We assume that the local prior $\pi ( \la _i )$ satisfies the following two conditions: 
\begin{itemize}
\item
There exist $c > 0$ and $C > 0$ such that 
\begin{align}
t \pi (t) \sim {C \over \{ 1 + \log (1 + t) \} [1 + \log \{ 1 + \log (1 + t) \} ]^{1 + c}} \label{eq:a_new_1} 
\end{align}
as $t \to \infty $. 
\item
There exist $c > 0$ and $M > 0$ such that for all $t > 0$, 
\begin{align}
t \pi (t) &\le M {1 \over 1 + \log (1 + 1 / t)} {1 \over 1 + \log (1 + t)} \non \\
&\quad \times {1 \over [1 + \log \{ 1 + \log (1 + 1 / t) \} ]^{1 + c}} {1 \over [1 + \log \{ 1 + \log (1 + t) \} ]^{1 + c}} \text{.} \label{eq:a_new_2} 
\end{align}
\end{itemize}

Let $\pit ( \al , \bbet , \ga ) = \ga ^p \pi ( \al , \ga \bbet , \ga )$. 
Suppose that there exist $a_{\bbet } > 0$, $b_{\bbet } > 0$, and $M > 0$ such that for all $( \al , \bbet , \ga ) \in (0, \infty ) \times \mathbb{R} ^p \times (0, \infty )$, 
\begin{align}
\pit ( \al , \bbet , \ga ) &\le M \pi ( \al , \ga ) \prod_{k = 1}^{p} {| \bet _k |^{a_{\bbet } - 1} \over (1 + | \bet _k |)^{a_{\bbet } + b_{\bbet }}} \text{.} \label{eq:a3_new} 
\end{align}
Suppose that there exists $\ka > 0$ such that 
\begin{align}
E \Big[ {1 \over \ga ^{\ka }} + \ga ^{\ka } + {1 \over \{ ( \min \{ \al , \al ^{1 / 2} \} ) \ga \} ^{\ka }} \Big] < \infty \text{.} \non 
\end{align}

\begin{thm}
\label{thm:robustness_new} 
The posterior is robust to the outliers. 
That is, 
\begin{align}
\lim_{\om \to \infty } p( \al , \bbet , \ga | \mathcal{T} ) = p( \al , \bbet , \ga | \{ t_i | i \in \Kc \} ) \non 
\end{align}
for all $p( \al , \bbet , \ga ) \in (0, \infty ) \times \mathbb{R} ^p \times (0, \infty )$. 
\end{thm}

\begin{remark}
\label{rem:loglog} 
Let $a, b > 0$. 
Let 
\begin{align}
\pi _2 (w | a, b) = {1 \over B(a, b)} {1 \over 1 + w} {1 \over 1 + \log (1 + w)} {[ \log \{ 1 + \log (1 + w) \} ]^{a - 1} \over [1 + \log \{ 1 + \log (1 + w) \} ]^{a + b}} \label{eq:prior_new} 
\end{align}
for $w \in (0, \infty )$. 
Then $\int_{0}^{\infty } \pi _2 (w | a, b) dw = 1$ and 
\begin{align}
\pi _2 (w | a, b) &\propto \begin{cases} \displaystyle w^{a - 1} \text{,} & \text{as $w \to 0$} \text{,} \\ \displaystyle {1 \over w} {1 \over \log w} {1 \over ( \log \log w)^{1 + b}} \text{,} & \text{as $w \to \infty $} \text{.} \end{cases} \non \end{align}
Thus, the prior (\ref{eq:prior_new}) satisfies (\ref{eq:a_new_2}) as well as (\ref{eq:a_new_1}). 
(The prior becomes the DLH density of the main text when %
$a = 1$ and $b = c$.) 
\end{remark}

\subsection{Proof of Remark \ref{rem:loglog}}
Here, we prove the results in Remark \ref{rem:loglog}. 

\begin{proof}[Proof of Remark \ref{rem:loglog}]
First, $u \sim {\rm{SB}} (a, b)$ implies $\exp ( e^u - 1) - 1 \sim \pi _2 ( \cdot | a, b)$ and we have $\int_{0}^{\infty } \pi _2 (w | a, b) dw = 1$. 
Next, note that 
\begin{align}
&w {1 \over 1 + w} {1 \over 1 + \log (1 + w)} {[ \log \{ 1 + \log (1 + w) \} ]^{a - 1} \over [1 + \log \{ 1 + \log (1 + w) \} ]^{a + b}} \non \\
&\le w {1 \over [ \log \{ 1 + \log (1 + w) \} ]^{1 - a}} \le w M_1 {1 \over w^{1 - a}} = M_1 {1 \over 1 / w^a} \non \\
&\le  M_2 {1 \over 1 + \log (1 + 1 / w)} {1 \over 1 + \log \{ 1 + \log (1 + 1 / w) \} ^{1 + b}} \non 
\end{align}
for all $w \le 1$ for some $M_1 , M_2 > 0$ and that 
\begin{align}
&w \times {1 \over 1 + w} {1 \over 1 + \log (1 + w)} {[ \log \{ 1 + \log (1 + w) \} ]^{a - 1} \over [1 + \log \{ 1 + \log (1 + w) \} ]^{a + b}} \non \\
&\le M_3 {1 \over 1 + \log (1 + w)} {1 \over [1 + \log \{ 1 + \log (1 + w) \} ]^{1 + b}} \non 
\end{align}
for all $w \ge 1$ for some $M_3 > 0$. 
Then it follows that the prior (\ref{eq:prior_new}) satisfies (\ref{eq:a_new_2}). 
\end{proof}

\subsection{Proof of Theorem \ref{thm:robustness_new}}
Here, we prove Theorem \ref{thm:robustness_new}. 
First, note that 
\begin{align}
&p( \al , \bbe , \ga | \mathcal{T} ) \non \\
&\propto \pi ( \al , \bbe , \ga ) \prod_{i = 1}^{n} \int_{0}^{\infty } \pi ( \la _i ) {\al ^{\al } \ga \over \Ga ( \al )} {\{ \exp ( \x_i^{\top} \bbe ) \} ^{\al } \over {\la _i}^{\al }} {t_i}^{\ga \al - 1} e^{- \al \{ \exp ( \x_i^{\top} \bbe ) / \la _i \} {t_i}^{\ga }} d{\la _i} \non \\
&= \pi ( \al , \bbe , \ga ) \prod_{i = 1}^{n} \int_{0}^{\infty } {t_i}^{\ga } \exp ( \x_i^{\top} \bbe ) \pi ( {t_i}^{\ga } \exp ( \x_i^{\top} \bbe ) u) {\al ^{\al } \ga \over \Ga ( \al )} {{t_i}^{- 1} \over u^{\al }} e^{- \al / u} du \non \\
&\propto \pi ( \al , \bbe , \ga ) \prod_{i = 1}^{n} \int_{0}^{\infty } {\al ^{\al } \ga \over \Ga ( \al )} {e^{- \al / u} \over u^{1 + \al }} {{t_i}^{\ga } \exp ( \x_i^{\top} \bbe ) u \pi ( {t_i}^{\ga } \exp ( \x_i^{\top} \bbe ) u) \over C / ( \{ 1 + \log (1 + t_i ) \} [1 + \log \{ 1 + \log (1 + t_i ) \} ^{1 + c} ])} du \text{.} \non 
\end{align}
Then, by making the change of variables $\bbe = \ga \bbet $, we have 
\begin{align}
&p( \al , \bbet , \ga | \mathcal{T} ) \propto \pit ( \al , \bbet , \ga ) \prod_{i = 1}^{n} f( \al , \x_i^{\top} \bbet , \ga ; t_i ) \text{,} \non 
\end{align}
where 
\begin{align}
&f( \al , \x_i^{\top} \bbet , \ga ; t_i ) \non \\
&= \int_{0}^{\infty } {\al ^{\al } \ga \over \Ga ( \al )} {e^{- \al / u} \over u^{1 + \al }} {{t_i}^{\ga } \exp ( \ga \x_i^{\top} \bbet ) u \pi ( {t_i}^{\ga } \exp ( \ga \x_i^{\top} \bbet ) u) \over C / ( \{ 1 + \log (1 + t_i ) \} [1 + \log \{ 1 + \log (1 + t_i ) \} ^{1 + c} ])} du \non 
\end{align}
for $i = 1, \dots , n$. 
It follows from Lemma \ref{lem:pointwise} that the posterior is robust if \begin{align}
&\lim_{\om \to \infty } \int_{(0, \infty ) \times \mathbb{R} ^p \times (0, \infty )} \pit ( \al , \bbet , \ga ) \Big\{ \prod_{i = 1}^{n} f( \al , \x_i^{\top} \bbet , \ga ; t_i ) \Big\} d( \al , \bbet , \ga ) \non \\
&= \int_{(0, \infty ) \times \mathbb{R} ^p \times (0, \infty )} \pit ( \al , \bbet , \ga ) \Big\{ \prod_{i \in \Kc } f( \al , \x_i^{\top} \bbet , \ga ; t_i ) \Big\} d( \al , \bbet , \ga ) \text{.} \label{eq:NC_limit} 
\end{align}

\begin{proof}[Proof of (\ref{eq:NC_limit})]
Fix $N > 0$ (for example, let $N = 1$). 
By Lemma 1 of \cite{hamura2024posterior}, there exist $\de > 0$, $0 < \ep < 1$, and $\underline{\om } > N / \de $ such that for all $\om \ge \underline{\om }$, 
\begin{align}
\mathbb{R} ^p &\subset \Big( \bigcap_{i \in \Lc } \{ \bbe \in \mathbb{R} ^p | | \log t_i + \x_i^{\top} \bbe | > \ep | \log t_i | \} \Big) \non \\
&\quad \cup \bigcup_{l = 1}^{\min \{ p, | \Lc | \} } \bigcup_{\substack{i_1 , \dots , i_l \in \Lc \\ i_1 < \dots < i_l}} \Big( \Big( \bigcap_{i \in \Lc \setminus \{ i_1 , \dots , i_l \} } \{ \bbet \in \mathbb{R} ^p | | \log t_i + \x_i^{\top} \bbe | > \ep | \log t_i | \} \Big) \non \\
&\quad \cap \Big( \bigcup_{1 \le k_1 < \dots < k_l \le p} \bigcap_{k \in \{ k_1 , \dots , k_l \} } \{ \bbet \in \mathbb{R} ^p | | \bet _k | \ge \de \om \} \Big) \non \\
&\quad \cap \Big[ \Big( \bigcap_{j \in \Jh } \{ \bbet \in \mathbb{R} ^p | | a_j + {\x _j}^{\top } \bbe | > N \} \Big) \non \\
&\quad \cup \bigcup_{1 \le q \le p - l} \bigcup_{\substack{j_1 , \dots , j_q \in \Jh \\ j_1 < \dots < j_q}} \Big\{ \Big( \bigcap_{j \in \{ j_1 , \dots , j_q \} } \{ \bbet \in \mathbb{R} ^p | | a_j + {\x _j}^{\top } \bbet | \le N \} \Big) \non \\
&\quad \cap \bigcap_{j \in \Jh \setminus \{ j_1 , \dots , j_q \} } \{ \bbet \in \mathbb{R} ^p | | a_j + {\x _j}^{\top } \bbet | > N \} \Big\} \Big] \Big) \text{,} \label{trobustness_newp1} 
\end{align}
where $\Jh = \{ - 1, \dots , - p \} \cup \Kc $ and where $a_j = 0$ and $\x _j = \e _{- j}^{(p)}$ for $j = - 1, \dots , - p$. 
In this proof, we assume that $\underline{\om }$ (and hence $\om $) is sufficiently large. 

By Lemma \ref{lem:L_2_new}, 
\begin{align}
&1(| \log t_i + \x_i^{\top} \bbet | > \ep | \log t_i |) f( \al , \x_i^{\top} \bbet , \ga ; t_i ) \non \\
&\le M_1 1(| \log t_i + \x_i^{\top} \bbet | > \ep | \log t_i |) \non \\
&\quad \times \int_{0}^{\infty } \Big\{ {\al ^{\al } \ga \over \Ga ( \al )} {e^{- \al / u} \over u^{1 + \al }} {1 + \log (1 + t_i ) \over 1 + \log \{ 1 + ( {t_i}^{\ep } \min \{ u^{1 / \ga } , 1 / u^{1 / \ga } \} )^{\ga } \} } \non \\
&\quad \times {[1 + \log \{ 1 + \log (1 + t_i ) \} ]^{1 + c} \over (1 + \log [1 + \log \{ 1 + ( {t_i}^{\ep } \min \{ u^{1 / \ga } , 1 / u^{1 / \ga } \} )^{\ga } \} ])^{1 + c}} \Big\} du \non 
\end{align}
for all $i \in \Lc $ for some $M_1 > 0$. 
If $\ga > 1 / \om ^{1 / 2}$, 
\begin{align}
&\{ u \in (0, \infty ) | \min \{ u^{1 / \ga } , 1 / u^{1 / \ga } \} \le 1/ {t_i}^{\ep / 2} \} \non \\
&\subset \{ u \in (0, \infty ) | \min \{ u, 1 / u \} \le M_2 / \exp ( \om ^{1 / 2} / M_3 ) \} \non 
\end{align}
for some $M_2 , M_3 > 0$. 
If $\al \ge 1 / \om ^{1 / 4}$, then $0 < \rho ( \om ) \le ( \al / 2) / (1 + \al ) < 1$, where $\rho ( \om ) = M_2 / \exp ( \om ^{1 / 2} / M_3 )$, and therefore 
\begin{align}
(0, \rho ( \om )] &\subset \Big\{ u \in (0, \infty ) \Big| {e^{- \al / u} \over u^{1 + \al }} \le e^{- \al / \{ 2 \rho ( \om ) \} } \Big\{ {e^{- \al / (2 u)} \over u^{1 + \al }} \Big\} \Big| _{u = \rho ( \om )} \Big\} \non \\
&\subset \Big\{ u \in (0, \infty ) \Big| {e^{- \al / u} \over u^{1 + \al }} \le e^{- \al / \{ 2 \rho ( \om ) \} } \exp \Big\{ \log {1 \over \rho ( \om )} - {\al \over 4} {1 \over \rho ( \om )} \Big\} \Big\} \non \\
&\subset \Big\{ u \in (0, \infty ) \Big| {e^{- \al / u} \over u^{1 + \al }} \le e^{- \al / \{ 2 \rho ( \om ) \} } \Big\} \non 
\end{align}
Thus, 
\begin{align}
&1( \al \ge 1 / \om ^{1 / 4} ) 1( \ga > 1 / \om ^{1 / 2} ) \int_{0}^{\infty } \Big\{ {\al ^{\al } \ga \over \Ga ( \al )} {e^{- \al / u} \over u^{1 + \al }} {1 + \log (1 + t_i ) \over 1 + \log \{ 1 + ( {t_i}^{\ep } \min \{ u^{1 / \ga } , 1 / u^{1 / \ga } \} )^{\ga } \} } \non \\
&\quad \times {[1 + \log \{ 1 + \log (1 + t_i ) \} ]^{1 + c} \over (1 + \log [1 + \log \{ 1 + ( {t_i}^{\ep } \min \{ u^{1 / \ga } , 1 / u^{1 / \ga } \} )^{\ga } \} ])^{1 + c}} \Big\} du \non \\
&\le \int_{0}^{\rho ( \om )} {\al ^{\al } \ga \over \Ga ( \al )} {e^{- \al / u} \over u^{1 + \al }} \{ 1 + \log (1 + t_i ) \} [1 + \log \{ 1 + \log (1 + t_i ) \} ]^{1 + c} du \non \\
&\quad + \int_{\rho ( \om )}^{1 / \rho ( \om )} {\al ^{\al } \ga \over \Ga ( \al )} {e^{- \al / u} \over u^{1 + \al }} {1 + \log (1 + t_i ) \over 1 + \log \{ 1 + ( {t_i}^{\ep / 2} )^{\ga } \} } {[1 + \log \{ 1 + \log (1 + t_i ) \} ]^{1 + c} \over (1 + \log [1 + \log \{ 1 + ( {t_i}^{\ep / 2} )^{\ga } \} ])^{1 + c}} du \non \\
&\quad + 1( \al \ge 1 / \om ^{1 / 4} ) \int_{1 / \rho ( \om )}^{\infty } {\al ^{\al } \ga \over \Ga ( \al )} {1 \over u^{1 + \al }} \{ 1 + \log (1 + t_i ) \} [1 + \log \{ 1 + \log (1 + t_i ) \} ]^{1 + c} du \non \\
&\le {\al ^{\al } \ga \over \Ga ( \al )} e^{- \al / \{ 2 \rho ( \om ) \} } \{ 1 + \log (1 + t_i ) \} [1 + \log \{ 1 + \log (1 + t_i ) \} ]^{1 + c} {M_2 \over \exp ( \om ^{1 / 2} / M_3 )} \non \\
&\quad + M_4 \{ 1 + \log (1 + 1 / \ga ) \} \non \\
&\quad + 1( \al \ge 1 / \om ^{1 / 4} ) {\al ^{\al } \ga \over \Ga ( \al )} {\{ \rho ( \om ) \} ^{\al } \over \al } \{ 1 + \log (1 + t_i ) \} [1 + \log \{ 1 + \log (1 + t_i ) \} ]^{1 + c} \non \\
&\le M_5 \Big[ {\ga \over \exp \{ \om ^{1 / 2} / (2 M_3 ) \} } + 1 + \log (1 + 1 / \ga ) + \{ 1( \al \ge 1) e^{\al } \exp ( - \al \om ^{1 / 2} / M_3 ) \non \\
&\quad + 1(1 / \om ^{1 / 4} \le \al \le 1) \exp ( - \al \om ^{1 / 2} / M_3 ) \} \ga \om ( \log \om )^{1 + c} \Big] \non 
\end{align}
for some $M_4 , M_5 > 0$ and 
\begin{align}
&1( \al \ge 1 / \om ^{1 / 4} ) 1(| \log t_i + \x_i^{\top} \bbet | > \ep | \log t_i |) 1(1 / \om ^{1 / 2} \le \ga \le \exp \{ \om ^{1 / 4} / (2 M_3 ) \} ) \non \\
&\le M_6 \{ 1 + \log (1 + 1 / \ga ) \} / f( \al , \x_i^{\top} \bbet , \ga ; t_i ) \non 
\end{align}
for all $i \in \Lc $ for some $M_6 > 0$. 
Hence, by the dominated convergence theorem, 
\begin{align}
&\lim_{\om \to \infty } \int_{(0, \infty ) \times \mathbb{R} ^p \times (0, \infty )} \Big[ 1( \al \ge 1 / \om ^{1 / 4} ) \Big\{ \prod_{i \in \Lc } 1(| \log t_i + \x_i^{\top} \bbet | > \ep | \log t_i |) \Big\} \non \\
&\quad \times 1(1 / \om ^{1 / 2} \le \ga \le \exp \{ \om ^{1 / 4} / (2 M_3 ) \} ) \pit ( \al , \bbet , \ga ) \prod_{i = 1}^{n} f( \al , \x_i^{\top} \bbet , \ga ; t_i ) \Big] d( \al , \bbet , \ga ) \non \\
&= \int_{(0, \infty ) \times \mathbb{R} ^p \times (0, \infty )} \Big( \lim_{\om \to \infty } \Big[ 1( \al \ge 1 / \om ^{1 / 4} ) \Big\{ \prod_{i \in \Lc } 1(| \log t_i + \x_i^{\top} \bbet | > \ep | \log t_i |) \Big\} \non \\
&\quad \times 1(1 / \om ^{1 / 2} \le \ga \le \exp \{ \om ^{1 / 4} / (2 M_3 ) \} ) \pit ( \al , \bbet , \ga ) \prod_{i = 1}^{n} f( \al , \x_i^{\top} \bbet , \ga ; t_i ) \Big] \Big) d( \al , \bbet , \ga ) \non \\
&= \int_{(0, \infty ) \times \mathbb{R} ^p \times (0, \infty )} \pit ( \al , \bbet , \ga ) \Big\{ \prod_{i \in \Kc } f( \al , \x_i^{\top} \bbet , \ga ; t_i ) \Big\} d( \al , \bbet , \ga ) \text{.} \label{trobustness_newp2} 
\end{align}

Fix $0 < \eta < \ka / (n + \ka )$. 
By Lemma \ref{lem:K_new}, 
\begin{align}
&\pit ( \al , \bbet , \ga ) \prod_{i = 1}^{n} f( \al , \x_i^{\top} \bbet , \ga ; t_i ) \non \\
&\le M_7 {\pit ( \al , \bbet , \ga ) \over \{ ( \min \{ \al , \al ^{1 / 2} \} ) \ga \} ^{\ka }} \{ \om ( \log \om )^{1 + c} \} ^{| \Lc |} \{ ( \min \{ \al , \al ^{1 / 2} \} ) \ga \} ^{n + \ka } \non \\
&\le M_8 {\pit ( \al , \bbet , \ga ) \over \{ ( \min \{ \al , \al ^{1 / 2} \} ) \ga \} ^{\ka }} {\{ \om ( \log \om )^{1 + c} \} ^{| \Lc |} \over \om ^{n + \ka - \eta (n + \ka )}} \le M_8 {\pit ( \al , \bbet , \ga ) \over \{ ( \min \{ \al , \al ^{1 / 2} \} ) \ga \} ^{\ka }} \label{trobustness_newp3} 
\end{align}
for some $M_7 , M_8 > 0$ if $( \ep / 4) ( \min \{ \al , \al ^{1 / 2} \} ) \ga \min_{i \in \Lc } \log t_i \le \om ^{\eta }$. 
By Lemma \ref{lem:K_new}, 
\begin{align}
&\pit ( \al , \bbet , \ga ) \prod_{i = 1}^{n} f( \al , \x_i^{\top} \bbet , \ga ; t_i ) \non \\
&\le M_9 \ga ^{\ka / 2} \pit ( \al , \bbet , \ga ) {\{ \om ( \log \om )^{1 + c} \} ^{| \Lc |} \over \ga ^{\ka / 2}} \{ ( \min \{ \al , \al ^{1 / 2} \} ) \ga \} ^n \non \\
&\le M_{10} \ga ^{\ka / 2} \pit ( \al , \bbet , \ga ) {\{ \om ( \log \om )^{1 + c} \} ^{| \Lc |} \over \exp \{ \ka \om ^{\eta } / (2 N) \} } \{ \log (1 + \ga ) \} ^n \le M_{11} ( \ga ^{\ka / 2} + \ga ^{\ka }) \pit ( \al , \bbet , \ga ) \label{trobustness_newp4} 
\end{align}
for some $M_9 , M_{10} , M_{11} > 0$ if $\ga \ge \exp ( \om ^{\eta } / N)$ and $( \ep / 4) ( \min \{ \al , \al ^{1 / 2} \} ) \ga \le \log (1 + \ga )$. 
Fix $0 \le l \le \min \{ p, | \Lc | \} $ and $0 \le q \le p - l$. 
Fix $i_1 , \dots , i_l \in \Lc $ with $i_1 < \dots < i_l$ and fix $1 \le k_1 < \dots < k_l \le p$. 
Fix $j_1 , \dots , j_q \in \Jh $ with $j_1 < \dots < j_q$. 
By (\ref{trobustness_newp1}), (\ref{trobustness_newp2}), (\ref{trobustness_newp3}), and (\ref{trobustness_newp4}), it is sufficient to prove that 
\begin{align}
&\lim_{\om \to \infty } \int_{(0, \infty ) \times \mathbb{R} ^p \times (0, \infty )} h( \al , \bbet , \ga ; \om ) d( \al , \bbet , \ga ) = 0 \text{,} \non 
\end{align}
where 
\begin{align}
h( \al , \bbet , \ga ; \om ) &= 1 \Big( \bbet \in \Big( \bigcap_{i \in \Lc \setminus \{ i_1 , \dots , i_l \} } \{ \bbet \in \mathbb{R} ^p | | \log t_i + \x_i^{\top} \bbet | > \ep | \log t_i | \} \Big) \non \\
&\quad \cap \Big( \bigcap_{k \in \{ k_1 , \dots , k_l \} } \{ \bbet \in \mathbb{R} ^p | | \bet _k | \ge \de \om \} \Big) \non \\
&\quad \cap \Big( \bigcap_{j \in \{ j_1 , \dots , j_q \} } \{ \bbet \in \mathbb{R} ^p | | a_j + {\x _j}^{\top } \bbet | \le N \} \Big) \non \\
&\quad \cap \bigcap_{j \in \Jh \setminus \{ j_1 , \dots , j_q \} } \{ \bbet \in \mathbb{R} ^p | | a_j + {\x _j}^{\top } \bbet | > N \} \Big) \non \\
&\quad \times \{ 1 - 1( \ga \ge \exp ( \om ^{\eta } / N)) 1(( \ep / 4) ( \min \{ \al , \al ^{1 / 2} \} ) \ga \le \log (1 + \ga )) \} \non \\
&\quad \times 1(( \ep / 4) ( \min \{ \al , \al ^{1 / 2} \} ) \ga \min_{i \in \Lc } \log t_i \ge \om ^{\eta } ) \non \\
&\quad \times \Big[ 1 - 1( \al \ge 1 / \om ^{1 / 4} ) \Big\{ \prod_{i \in \Lc } 1(| \log t_i + \x_i^{\top} \bbet | > \ep | \log t_i |) \Big\} \non \\
&\quad \times 1(1 / \om ^{1 / 2} \le \ga \le \exp \{ \om ^{1 / 4} / (2 M_3 ) \} ) \Big] \non \\
&\quad \times \pit ( \al , \bbet , \ga ) \prod_{i = 1}^{n} f( \al , \x_i^{\top} \bbet , \ga ; t_i ) \text{.} \non 
\end{align}
By Lemma \ref{lem:L_improved}, we need only consider the case $l \ge 1$. 

There exist $q_1 , q_2 \ge 0$, $k_1' , \dots , k_{q_1}' = - 1, \dots , - p$, and $i_1' , \dots , i_{q_2}' \in \Kc $ such that $q_1 + q_2 = q$, $k_1' < \dots < k_{q_1}'$, $i_1' < \dots < i_{q_2}'$, and $\{ k_1' , \dots , k_{q_1}' \} \cup \{ i_1' , \dots , i_{q_2}' \} = \{ j_1 , \dots , j_q \} $. 
Since $\de \om > N$, we can assume that $\{ k_1 , \dots , k_l \} \cap \{ - k_1' , \dots , - k_{q_1}' \} = \emptyset $. 
Let $k_1'' , \dots , k_{p - l - q}'' \in \{ 1, \dots , p \} \setminus ( \{ k_1 , \dots , k_l \} \cup \{ - k_1' , \dots , - k_{q_1}' \} )$ be such that $k_1'' < \dots < k_{p - l - q}''$. 
By (\ref{eq:a3_new}), 
\begin{align}
h( \al , \bbet , \ga ; \om ) %
&\le 1 \Big( \bbet \in \Big( \bigcap_{i \in \Lc \setminus \{ i_1 , \dots , i_l \} } \{ \bbet \in \mathbb{R} ^p | | \log t_i + \x_i^{\top} \bbet | > \ep | \log t_i | \} \Big) \cap \non \\
&\quad \Big( \bigcap_{k \in \{ k_1 , \dots , k_l \} } \{ \bbet \in \mathbb{R} ^p | | \bet _k | \ge \de \om \} \Big) \non \\
&\quad \cap \Big( \bigcap_{j \in \{ j_1 , \dots , j_q \} } \{ \bbet \in \mathbb{R} ^p | | a_j + {\x _j}^{\top } \bbet | \le N \} \Big) \cap \non \\
&\quad \bigcap_{j \in \Jh \setminus \{ j_1 , \dots , j_q \} } \{ \bbet \in \mathbb{R} ^p | | a_j + {\x _j}^{\top } \bbet | > N \} \Big) \non \\
&\quad \times \{ 1 - 1( \ga \ge \exp ( \om ^{\eta } / N)) 1(( \ep / 4) ( \min \{ \al , \al ^{1 / 2} \} ) \ga \le \log (1 + \ga )) \} \non \\
&\quad \times 1(( \ep / 4) ( \min \{ \al , \al ^{1 / 2} \} ) \ga \min_{i \in \Lc } \log t_i \ge \om ^{\eta } ) \non \\
&\quad \times \Big\{ \prod_{k \in \{ - k_1' , \dots , - k_{q_1}' \} } {| \bet _k |^{a_{\bbet } - 1} \over (1 + | \bet _k |)^{a_{\bbet } + b_{\bbet }}} \Big\} \non \\
&\quad \times \Big\{ \prod_{i \in \{ i_1' , \dots , i_{q_2}' \} } f( \al , \x_i^{\top} \bbet , \ga ; t_i ) \Big\} \Big\{ \prod_{i \in \{ i_1 , \dots , i_l \} } f( \al , \x_i^{\top} \bbet , \ga ; t_i ) \Big\} \non \\
&\quad \times M \pi ( \al , \ga ) \Big\{ \prod_{k \in \{ 1, \dots , p \} \setminus ( \{ k_1 , \dots , k_l \} \cup \{ - k_1' , \dots , - k_{q_1}' \} )} {| \bet _k |^{a_{\bbet } - 1} \over (1 + | \bet _k |)^{a_{\bbet } + b_{\bbet }}} \Big\} \non \\
&\quad \times \Big\{ \prod_{k \in \{ k_1 , \dots , k_l \} } {| \bet _k |^{a_{\bbet } - 1} \over (1 + | \bet _k |)^{a_{\bbet } + b_{\bbet }}} \Big\} \non \\
&\quad \times \Big\{ \prod_{i \in \Kc \setminus \{ i_1' , \dots , i_{q_2}' \} } f( \al , \x_i^{\top} \bbet , \ga ; t_i ) \Big\} \Big\{ \prod_{i \in \Lc \setminus \{ i_1 , \dots , i_l \} } f( \al , \x_i^{\top} \bbet , \ga ; t_i ) \Big\} \text{.} \non 
\end{align}
By Lemma \ref{lem:L_2_new}, 
\begin{align}
h( \al , \bbet , \ga ; \om ) &\le M_{12} 1 \Big( \bbet \in \Big( \bigcap_{i \in \Lc \setminus \{ i_1 , \dots , i_l \} } \{ \bbet \in \mathbb{R} ^p | | \log t_i + \x_i^{\top} \bbet | > \ep | \log t_i | \} \Big) \non \\
&\quad \cap \Big( \bigcap_{k \in \{ k_1 , \dots , k_l \} } \{ \bbet \in \mathbb{R} ^p | | \bet _k | \ge \de \om \} \Big) \non \\
&\quad \cap \Big( \bigcap_{j \in \{ j_1 , \dots , j_q \} } \{ \bbet \in \mathbb{R} ^p | | a_j + {\x _j}^{\top } \bbet | \le N \} \Big) \non \\
&\quad \cap \bigcap_{j \in \Jh \setminus \{ j_1 , \dots , j_q \} } \{ \bbet \in \mathbb{R} ^p | | a_j + {\x _j}^{\top } \bbet | > N \} \Big) \non \\
&\quad \times \{ 1 - 1( \ga \ge \exp ( \om ^{\eta } / N)) 1(( \ep / 4) ( \min \{ \al , \al ^{1 / 2} \} ) \ga \le \log (1 + \ga )) \} \non \\
&\quad \times 1(( \ep / 4) ( \min \{ \al , \al ^{1 / 2} \} ) \ga \min_{i \in \Lc } \log t_i \ge \om ^{\eta } ) \non \\
&\quad \times \Big\{ \prod_{k \in \{ - k_1' , \dots , - k_{q_1}' \} } {| \bet _k |^{a_{\bbet } - 1} \over (1 + | \bet _k |)^{a_{\bbet } + b_{\bbet }}} \Big\} \Big\{ \prod_{i \in \{ i_1' , \dots , i_{q_2}' \} } f( \al , \x_i^{\top} \bbet , \ga ; t_i ) \Big\} \non \\
&\quad \times \Big( \{ \om ( \log \om )^{1 + c} \} ^l \non \\
&\quad \times \prod_{i \in \{ i_1 , \dots , i_l \} } \int_{0}^{\infty } \Big[ {\al ^{\al } \ga \over \Ga ( \al )} {e^{- \al / u} \over u^{1 + \al }} {1 \over 1 + \log [1 + \exp \{ | \ga (\log t_i + \x_i^{\top} \bbet ) + \log u| \} ]} \non \\
&\quad \times {1 \over \{ 1 + \log (1 + \log [1 + \exp \{ | \ga ( \log t_i + \x_i^{\top} \bbet ) + \log u| \} ]) \} ^{1 + c}} \Big] du \Big) \non \\
&\quad \times \Big\{ \prod_{k \in \{ k_1'' , \dots , k_{p - l - q}'' \} } {| \bet _k |^{a_{\bbet } - 1} \over (1 + | \bet _k |)^{a_{\bbet } + b_{\bbet }}} \Big\} \pi ( \al , \ga ) {1 \over \om ^{l (1 + b_{\bbet } )}} \non \\
&\quad \times \Big\{ \prod_{i \in \Kc \setminus \{ i_1' , \dots , i_{q_2}' \} } f( \al , \x_i^{\top} \bbet , \ga ; t_i ) \Big\} \Big\{ \prod_{i \in \Lc \setminus \{ i_1 , \dots , i_l \} } f( \al , \x_i^{\top} \bbet , \ga ; t_i ) \Big\} \non 
\end{align}
for some $M_{12} > 0$. 
By Lemmas \ref{lem:K_improved} and \ref{lem:L_improved}, 
\begin{align}
h( \al , \bbet , \ga ; \om ) &\le M_{13} 1 \Big( \bbet \in \Big( \bigcap_{i \in \Lc \setminus \{ i_1 , \dots , i_l \} } \{ \bbet \in \mathbb{R} ^p | | \log t_i + \x_i^{\top} \bbet | > \ep | \log t_i | \} \Big) \non \\
&\quad \cap \Big( \bigcap_{k \in \{ k_1 , \dots , k_l \} } \{ \bbet \in \mathbb{R} ^p | | \bet _k | \ge \de \om \} \Big) \non \\
&\quad \cap \Big( \bigcap_{j \in \{ j_1 , \dots , j_q \} } \{ \bbet \in \mathbb{R} ^p | | a_j + {\x _j}^{\top } \bbet | \le N \} \Big) \non \\
&\quad \cap \bigcap_{j \in \Jh \setminus \{ j_1 , \dots , j_q \} } \{ \bbet \in \mathbb{R} ^p | | a_j + {\x _j}^{\top } \bbet | > N \} \Big) \non \\
&\quad \times \{ 1 - 1( \ga \ge \exp ( \om ^{\eta } / N)) 1(( \ep / 4) ( \min \{ \al , \al ^{1 / 2} \} ) \ga \le \log (1 + \ga )) \} \non \\
&\quad \times 1(( \ep / 4) ( \min \{ \al , \al ^{1 / 2} \} ) \ga \min_{i \in \Lc } \log t_i \ge \om ^{\eta } ) \non \\
&\quad \times \Big\{ \prod_{k \in \{ - k_1' , \dots , - k_{q_1}' \} } {| \bet _k |^{a_{\bbet } - 1} \over (1 + | \bet _k |)^{a_{\bbet } + b_{\bbet }}} \Big\} \Big\{ \prod_{i \in \{ i_1' , \dots , i_{q_2}' \} } f( \al , \x_i^{\top} \bbet , \ga ; t_i ) \Big\} \non \\
&\quad \times \Big( \{ \om ( \log \om )^{1 + c} \} ^l \non \\
&\quad \times \prod_{i \in \{ i_1 , \dots , i_l \} } \int_{0}^{\infty } \Big[ {\al ^{\al } \ga \over \Ga ( \al )} {e^{- \al / u} \over u^{1 + \al }} {1 \over 1 + \log [1 + \exp \{ | \ga (\log t_i + \x_i^{\top} \bbet ) + \log u| \} ]} \non \\
&\quad \times {1 \over \{ 1 + \log (1 + \log [1 + \exp \{ | \ga ( \log t_i + \x_i^{\top} \bbet ) + \log u| \} ]) \} ^{1 + c}} \Big] du \Big) \non \\
&\quad \times \Big\{ \prod_{k \in \{ k_1'' , \dots , k_{p - l - q}'' \} } {| \bet _k |^{a_{\bbet } - 1} \over (1 + | \bet _k |)^{a_{\bbet } + b_{\bbet }}} \Big\} \pi ( \al , \ga ) {1 \over \om ^{l (1 + b_{\bbet } )}} \non \\
&\quad \times \{ 1 + \log (1 + \ga ) \} ^{| \Kc | - q_2} \prod_{i \in \Lc \setminus \{ i_1 , \dots , i_l \} } [1 + \{ 1 + \log (1 + 1 / \ga ) \} ^{1 + c} ] \non 
\end{align}
for some $M_{13} > 0$. 
Therefore, 
\begin{align}
&\int_{(0, \infty ) \times \mathbb{R} ^p \times (0, \infty )} h( \al , \bbet , \ga ; \om ) d( \al , \bbet , \ga ) \non \\
&\le M_{13} \int_{(0, \infty ) \times \mathbb{R} ^p \times (0, \infty )} \Big\{ \Big\{ \prod_{k \in \{ - k_1' , \dots , - k_{q_1}' \} } {| \bet _k |^{a_{\bbet } - 1} \over (1 + | \bet _k |)^{a_{\bbet } + b_{\bbet }}} \Big\} \Big\{ \prod_{i \in \{ i_1' , \dots , i_{q_2}' \} } f( \al , \x_i^{\top} \bbet , \ga ; t_i ) \Big\} \non \\
&\quad \times \Big( \prod_{i \in \{ i_1 , \dots , i_l \} } \int_{0}^{\infty } \Big[ {\al ^{\al } \ga \over \Ga ( \al )} {e^{- \al / u} \over u^{1 + \al }} {1 \over 1 + \log [1 + \exp \{ | \ga (\log t_i + \x_i^{\top} \bbet ) + \log u| \} ]} \non \\
&\quad \times {1 \over \{ 1 + \log (1 + \log [1 + \exp \{ | \ga ( \log t_i + \x_i^{\top} \bbet ) + \log u| \} ]) \} ^{1 + c}} \Big] du \Big) \non \\
&\quad \times \Big\{ \prod_{k \in \{ k_1'' , \dots , k_{p - l - q}'' \} } {| \bet _k |^{a_{\bbet } - 1} \over (1 + | \bet _k |)^{a_{\bbet } + b_{\bbet }}} \Big\} \pi ( \al , \ga ) {( \log \om )^{l (1 + c)} \over \om ^{l b_{\bbet }}} \non \\
&\quad \times \{ 1 + \log (1 + \ga ) \} ^{| \Kc | - q_2} \Big( \prod_{i \in \Lc \setminus \{ i_1 , \dots , i_l \} } [1 + \{ 1 + \log (1 + 1 / \ga ) \} ^{1 + c} ] \Big) \Big\} d( \al , \bbet , \ga ) \non \\
&\le M_{14} {( \log \om )^{l (1 + c)} \over \om ^{l b_{\bbet }}} \int_{(0, \infty )^2} \Big\{ \pi ( \al , \ga ) \{ 1 + \log (1 + \ga ) \} ^{| \Kc | - q_2} \non \\
&\quad \times \Big( \prod_{i \in \Lc \setminus \{ i_1 , \dots , i_l \} } [1 + \{ 1 + \log (1 + 1 / \ga ) \} ^{1 + c} ] \Big) \Big\} d( \al , \ga ) \non 
\end{align}
for some $M_{14} > 0$, which converges to $0$ as $\om \to \infty $. 
This completes the proof. 
\end{proof}

\subsection{Lemmas}
Here, we prove lemmas.

\begin{lem}
\label{lem:ga} 
For any $x > 0$, 
\begin{align}
{x^{x - 1 / 2} \over \Ga (x) e^x} < {1 \over (2 \pi )^{1 / 2}} \text{.} \non 
\end{align}
\end{lem}

\begin{lem}
\label{lem:log} 
\hfill
\begin{itemize}
\item[{\rm{(i)}}]
For $v > 0$, 
\begin{align}
{1 + \log (1 + u) \over 1 + \log (1 + u v)} &\begin{cases} \to 1 \text{,} & \text{as $u \to \infty $} \text{,} \\ \le 1 + \log (1 + 1 / v) \text{,} & \text{for all $u > 0$} \text{.} \end{cases} \non 
\end{align}
\item[{\rm{(ii)}}]
For $v > 0$, 
\begin{align}
{1 + \log \{ 1 + \log (1 + u) \} \over 1 + \log \{ 1 + \log (1 + u v) \} } &\begin{cases} \to 1 \text{,} & \text{as $u \to \infty $} \text{,} \\ \le 1 + \log \{ 1 + \log (1 + 1 / v) \} \text{,} & \text{for all $u > 0$} \text{.} \end{cases} \non 
\end{align}
\end{itemize}
\end{lem}

\begin{lem}
\label{lem:pointwise} 
We have 
\begin{align}
&\lim_{t_i \to \infty } f( \al , \ga \x_i^{\top} \bbet , \ga ; t_i ) = 1 \text{.} \non 
\end{align}
\end{lem}

\begin{lem}
\label{lem:L_2_new} 
There exists $M > 0$ such that for all $\ep > 0$, all $i \in \Lc $, and all $( \al , \bbet , \ga ) \in (0, \infty ) \times \mathbb{R} ^p \times (0, \infty )$ with $| \log t_i + \x_i^{\top} \bbet | > \ep | \log t_i |$, we have 
\begin{align}
f( \al , \x_i^{\top} \bbet , \ga ; t_i ) &\le M \int_{0}^{\infty } \Big\{ {\al ^{\al } \ga \over \Ga ( \al )} {e^{- \al / u} \over u^{1 + \al }} {1 + \log (1 + t_i ) \over 1 + \log \{ 1 + ( {t_i}^{\ep } \min \{ u^{1 / \ga } , 1 / u^{1 / \ga } \} )^{\ga } \} } \non \\
&\quad \times {[1 + \log \{ 1 + \log (1 + t_i ) \} ]^{1 + c} \over (1 + \log [1 + \log \{ 1 + ( {t_i}^{\ep } \min \{ u^{1 / \ga } , 1 / u^{1 / \ga } \} )^{\ga } \} ])^{1 + c}} \Big\} du \text{.} \non 
\end{align}
Also, there exists $M > 0$ such that for all $i = 1, \dots , n$, 
\begin{align}
f( \al , \x_i^{\top} \bbet , \ga ; t_i ) &\le M \int_{0}^{\infty } \Big[ {\al ^{\al } \ga \over \Ga ( \al )} {e^{- \al / u} \over u^{1 + \al }} {1 + \log (1 + t_i ) \over 1 + \log [1 + \exp \{ | \ga (\log t_i + \x_i^{\top} \bbet ) + \log u| \} ]} \non \\
&\quad \times {[1 + \log \{ 1 + \log (1 + t_i ) \} ]^{1 + c} \over \{ 1 + \log (1 + \log [1 + \exp \{ | \ga ( \log t_i + \x_i^{\top} \bbet ) + \log u| \} ]) \} ^{1 + c}} \Big] du \non \\
&\le M \{ 1 + \log (1 + t_i ) \} [1 + \log \{ 1 + \log (1 + t_i ) \} ]^{1 + c} \ga \text{.} \non 
\end{align}
\end{lem}

\begin{lem}
\label{lem:K_new} 
\noindent
There exists $M > 0$ such that for all $i = 1, \dots , n$, 
\begin{align}
f( \al , \x_i^{\top} \bbet , \ga ; t_i ) &\le \begin{cases} M ( \min \{ \al , \al ^{1 / 2} \} ) \ga \text{,} & \text{if $i \in \Kc$} \text{,} \\ M \om ( \log \om )^{1 + c'} ( \min \{ \al , \al ^{1 / 2} \} ) \ga \text{,} & \text{if $i \in \Lc$} \text{.} \end{cases} \non 
\end{align}
\end{lem}

\begin{lem}
\label{lem:L} 
\hfill
\begin{itemize}
\item[{\rm{(i)}}]
For all $\De > 0$, 
\begin{align}
\int_{0}^{1 / (1 + \De )} {\al ^{\al } \over \Ga ( \al )} {e^{- \al / u} \over u^{1 + \al }} du &\le \int_{1}^{\infty } \al ^{1 / 2} \De \exp \Big\{ - {t^2 \over 4} ( \al ^{1 / 2} \De )^2 \Big\} dt \non \\
&\quad + \int_{1}^{\infty } ( \al \De )^{1 / 2} \exp \Big( - {t \over 4} \al \De \Big) dt \text{.} \non 
\end{align}
\item[{\rm{(ii)}}]
For all $\De > 0$, 
\begin{align}
&\int_{1 + \De }^{K} {\al ^{\al } \over \Ga ( \al )} {e^{- \al / u} \over u^{1 + \al }} du \le K \int_{- \infty }^{- 1} |t| \al ^{1 / 2} {\log (1 + \De ) \over 1 + \log K} \exp \Big[ - {|t|^2 \over 2} \al {\{ \log (1 + \De ) \} ^2 \over (1 + \log K)^2} \Big] dt \non 
\end{align}
for all $K > 1 + \De $. 
\item[{\rm{(iii)}}]
For all $K > 0$, 
\begin{align}
\int_{K}^{\infty } {\al ^{\al } \over \Ga ( \al )} {e^{- \al / u} \over u^{1 + \al }} du &\le \min \Big\{ {\al ^{\al } \over \Ga ( \al + 1)} {1 \over K^{\al }}, {1 \over \al ^{1 / 2}} {1 \over (K / e)^{\al }} \Big\} \text{.} \non 
\end{align}
\end{itemize}
\end{lem}

\begin{lem}
\label{lem:K_improved} 
Let $N > 0$. 
Let $i \in \Kc $. 
Suppose that $| \log t_i + \x_i^{\top} \bbet | \ge N$. 
Then there exists $M > 0$ such that 
\begin{align}
f( \al , \x_i^{\top} \bbet , \ga ; t_i ) &\le M \{ 1 + \log (1 + \ga ) \} \text{.} \non 
\end{align}
\end{lem}

\begin{lem}
\label{lem:L_improved} 
Let $\ep > 0$. 
Let $i \in \Lc $. 
Suppose that $| \log t_i + \x_i^{\top} \bbet | > \ep | \log t_i |$. 
Suppose that $( \ep / 4) ( \min \{ \al , \al ^{1 / 2} \} ) \ga \log t_i \ge \om ^{\eta }$. 
Suppose that we have $\al \ga \ge \log (1 + \ga )$ or $\ga \le \exp ( \om ^{\eta } / 4)$. 
Then there exists $M > 0$ such that 
\begin{align}
f( \al , \x_i^{\top} \bbet , \ga ; t_i ) \le M [1 + \{ 1 + \log (1 + 1 / \ga ) \} ^{1 + c} ] \text{.} \non 
\end{align}
\end{lem}

\begin{proof}[Proof of Lemma \ref{lem:ga}]
See the Supplementary Material of \cite{hohs2022}. 
\end{proof}

\begin{proof}[Proof of Lemma \ref{lem:log}]
See Lemma S1 of \cite{Hamura2022log} and Lemma S2 of \cite{hamura2024robust} for part (i). 
Since for all $u, v > 0$, 
\begin{align}
{1 + \log \{ 1 + \log (1 + u) \} \over 1 + \log \{ 1 + \log (1 + u v) \} } &= 1 + {\log [ \{ 1 + \log (1 + u) \} / \{ 1 + \log (1 + u v) \} ] \over 1 + \log \{ 1 + \log (1 + u v) \} } \non \\
&\le 1 + {\log \{ 1 + \log (1 + 1 / v) \} \over 1 + \log \{ 1 + \log (1 + u v) \} } \non \\
&\le 1 + \log \{ 1 + \log (1 + 1 / v) \} \text{,} \non 
\end{align}
part (ii) follows. 
\end{proof}

\begin{proof}[Proof of Lemma \ref{lem:pointwise}]
When $c' = c$ and $C' = C$, by Lemma \ref{lem:log}, 
\begin{align}
&{\al ^{\al } \ga \over \Ga ( \al )} {e^{- \al / u} \over u^{1 + \al }} {{t_i}^{\ga } \exp ( \ga \x_i^{\top} \bbet ) u \pi ( {t_i}^{\ga } \exp ( \ga \x_i^{\top} \bbet ) u) \over C / ( \{ 1 + \log (1 + t_i ) \} [1 + \log \{ 1 + \log (1 + t_i ) \} ]^{1 + c} )} \non \\
&\le {\al ^{\al } \ga \over \Ga ( \al )} {e^{- \al / u} \over u^{1 + \al }} \{ 1 + \log (1 + t_i ) \} [1 + \log \{ 1 + \log (1 + t_i ) \} ]^{1 + c} \non \\
&\quad \times {M \over C} {1 \over 1 + \log [1 + 1 / \{ {t_i}^{\ga } \exp ( \ga \x_i^{\top} \bbet ) u \} ]} {1 \over 1 + \log [1 + \{ {t_i}^{\ga } \exp ( \ga \x_i^{\top} \bbet ) u \} ]} \non \\
&\quad \times {1 \over \{ 1 + \log (1 + \log [1 + 1 / \{ {t_i}^{\ga } \exp ( \ga \x_i^{\top} \bbet ) u \} ]) \} ^{1 + c}} \non \\
&\quad \times {1 \over \{ 1 + \log (1 + \log [1 + \{ {t_i}^{\ga } \exp ( \ga \x_i^{\top} \bbet ) u\} ]) \} ^{1 + c}} \non \\
&\le {\al ^{\al } \ga \over \Ga ( \al )} {e^{- \al / u} \over u^{1 + \al }} {M \over C} {1 + \log (1 + t_i ) \over 1 + \log [1 + \{ {t_i}^{\ga } \exp ( \ga \x_i^{\top} \bbet ) u \} ]} \non \\
&\quad \times {[1 + \log \{ 1 + \log (1 + t_i ) \} ]^{1 + c} \over \{ 1 + \log (1 + \log [1 + \{ {t_i}^{\ga } \exp ( \ga \x_i^{\top} \bbet ) u \} ]) \} ^{1 + c}} \non \\
&\le \widetilde{M} {\al ^{\al } \over \Ga ( \al )} {e^{- \al / u} \over u^{1 + \al }} \{ 1 + \log (1 + 1 / \ga ) \} ^{1 + c} \non \\
&\quad \times (1 + \log [1 + 1 / \{ \exp ( \ga \x_i^{\top} \bbet ) u \} ]) \{ 1 + \log (1 + \log [1 + 1 / \{ \exp ( \ga \x_i^{\top} \bbet ) u \} ]) \} ^{1 + c} \non \\
&\le \widetilde{M} {\al ^{\al } \over \Ga ( \al )} {e^{- \al / u} \over u^{1 + \al }} \{ 1 + \log (1 + 1 / \ga ) \} ^{1 + c} \non \\
&\quad \times [1 + \log \{ 1 + 1 / \exp ( \ga \x_i^{\top} \bbet ) \} ] (1 + \log [1 + \log \{ 1 + 1 / \exp ( \ga \x_i^{\top} \bbet ) \} ])^{1 + c} \non \\
&\quad \times \{ 1 + \log (1 + 1 / u) \} [1 + \log \{ 1 + \log (1 + 1 / u) \} ]^{1 + c} \non 
\end{align}
for all $u \in (0, \infty )$ for large $t_i > 0$ for all $i = 1, \dots , n$ for some $\widetilde{M} > 0$. 
Therefore, by the dominated convergence theorem, we have 
\begin{align}
&\lim_{t_i \to \infty } f( \al , \ga \x_i^{\top} \bbet , \ga ; t_i ) \non \\
&= \int_{0}^{\infty } \Big[ {\al ^{\al } \ga \over \Ga ( \al )} {e^{- \al / u} \over u^{1 + \al }} \lim_{t_i \to \infty } \Big\{ {1 + \log (1 + t_i ) \over 1 + \log \{ 1 + {t_i}^{\ga } \exp ( \ga \x_i^{\top} \bbet ) u \} } \non \\
&\quad \times {[1 + \log \{ 1 + \log (1 + t_i ) \} ]^{1 + c'} \over (1 + \log [1 + \log \{ 1 + {t_i}^{\ga } \exp ( \ga \x_i^{\top} \bbet ) u \} ])^{1 + c'}} \Big\} \Big] du = 1 \non 
\end{align}
for all $i = 1, \dots , n$. 
\end{proof}

\begin{proof}[Proof of Lemma \ref{lem:L_2_new}]
Let $\ep > 0$, $i = 1, \dots , n$, and $( \al , \bbet , \ga ) \in (0, \infty ) \times \mathbb{R} ^p \times (0, \infty )$. 
By assumption (\ref{eq:a_new_2}), 
\begin{align}
&f( \al , \x_i^{\top} \bbet , \ga ; t_i ) \non \\
&= \int_{0}^{\infty } {\al ^{\al } \ga \over \Ga ( \al )} {e^{- \al / u} \over u^{1 + \al }} {{t_i}^{\ga } \exp ( \ga \x_i^{\top} \bbet ) u \pi ( {t_i}^{\ga } \exp ( \ga \x_i^{\top} \bbet ) u) \over C / ( \{ 1 + \log (1 + t_i ) \} [1 + \log \{ 1 + \log (1 + t_i ) \} ]^{1 + c} )} du \non \\
&\le M' \int_{0}^{\infty } \Big[ {\al ^{\al } \ga \over \Ga ( \al )} {e^{- \al / u} \over u^{1 + \al }} \{ 1 + \log (1 + t_i ) \} [1 + \log \{ 1 + \log (1 + t_i ) \} ]^{1 + c} \non \\
&\quad \times {1 \over 1 + \log [1 + 1 / \{ {t_i}^{\ga } \exp ( \ga \x_i^{\top} \bbet ) u \} ]} {1 \over 1 + \log [1 + \{ {t_i}^{\ga } \exp ( \ga \x_i^{\top} \bbet ) u \} ]} \non \\
&\quad \times {1 \over \{ 1 + \log (1 + \log [1 + 1 / \{ {t_i}^{\ga } \exp ( \ga \x_i^{\top} \bbet ) u \} ]) \} ^{1 + c}} \non \\
&\quad \times {1 \over \{ 1 + \log (1 + \log [1 + \{ {t_i}^{\ga } \exp ( \ga \x_i^{\top} \bbet ) u \} ]) \} ^{1 + c}} \Big] du \non \\
&\le M' \int_{0}^{\infty } \Big\{ {\al ^{\al } \ga \over \Ga ( \al )} {e^{- \al / u} \over u^{1 + \al }} {1 + \log (1 + t_i ) \over 1 + \log \{ 1 + \exp ( \ga | \log t_i + \x_i^{\top} \bbet + \log u^{1 / \ga } |) \} } \non \\
&\quad \times {[1 + \log \{ 1 + \log (1 + t_i ) \} ]^{1 + c} \over (1 + \log [1 + \log \{ 1 + \exp ( \ga | \log t_i + \x_i^{\top} \bbet + \log u^{1 / \ga } |) \} ])^{1 + c}} \Big\} du \non 
\end{align}
for some $M' > 0$. 
Therefore, 
\begin{align}
&f( \al , \x_i^{\top} \bbet , \ga ; t_i ) \non \\
&\le M' \int_{0}^{\infty } \Big\{ {\al ^{\al } \ga \over \Ga ( \al )} {e^{- \al / u} \over u^{1 + \al }} {1 + \log (1 + t_i ) \over 1 + \log \{ 1 + \exp ( \ga | \log t_i + \x_i^{\top} \bbet | - | \log u| ) \} } \non \\
&\quad \times {[1 + \log \{ 1 + \log (1 + t_i ) \} ]^{1 + c} \over (1 + \log [1 + \log \{ 1 + \exp ( \ga | \log t_i + \x_i^{\top} \bbet | - | \log u|) \} ])^{1 + c}} \Big\} du \non \\
&= M' \int_{0}^{\infty } \Big\{ {\al ^{\al } \ga \over \Ga ( \al )} {e^{- \al / u} \over u^{1 + \al }} {1 + \log (1 + t_i ) \over 1 + \log \{ 1 + \exp ( \ga | \log t_i + \x_i^{\top} \bbet |) \min \{ u, 1 / u \} \} } \non \\
&\quad \times {[1 + \log \{ 1 + \log (1 + t_i ) \} ]^{1 + c} \over (1 + \log [1 + \log \{ 1 + \exp ( \ga | \log t_i + \x_i^{\top} \bbet |) \min \{ u, 1 / u \} \} ])^{1 + c}} \Big\} du \non \\
&= M' \int_{0}^{\infty } \Big[ {\al ^{\al } \ga \over \Ga ( \al )} {e^{- \al / u} \over u^{1 + \al }} {1 + \log (1 + t_i ) \over 1 + \log [1 + \{ \exp (| \log t_i + \x_i^{\top} \bbet |) \min \{ u^{1 / \ga } , 1 / u^{1 / \ga } \} \} ^{\ga } ]} \non \\
&\quad \times {[1 + \log \{ 1 + \log (1 + t_i ) \} ]^{1 + c} \over \{ 1 + \log (1 + \log [1 + \{ \exp (| \log t_i + \x_i^{\top} \bbet |) \min \{ u^{1 / \ga } , 1 / u^{1 / \ga } \} \} ^{\ga } ]) \} ^{1 + c}} \Big] du \non \\
&\le M' \int_{0}^{\infty } \Big\{ {\al ^{\al } \ga \over \Ga ( \al )} {e^{- \al / u} \over u^{1 + \al }} {1 + \log (1 + t_i ) \over 1 + \log \{ 1 + ( {t_i}^{\ep } \min \{ u^{1 / \ga } , 1 / u^{1 / \ga } \} )^{\ga } \} } \non \\
&\quad \times {[1 + \log \{ 1 + \log (1 + t_i ) \} ]^{1 + c} \over (1 + \log [1 + \log \{ 1 + ( {t_i}^{\ep } \min \{ u^{1 / \ga } , 1 / u^{1 / \ga } \} )^{\ga } \} ])^{1 + c}} \Big\} du \non 
\end{align}
if $| \log t_i + \x_i^{\top} \bbet | > \ep | \log t_i |$. 
\end{proof}

\begin{proof}[Proof of Lemma \ref{lem:K_new}]
We have 
\begin{align}
&f( \al , \x_i^{\top} \bbet , \ga ; t_i ) \non \\
&= \int_{0}^{\infty } {\al ^{\al } \ga \over \Ga ( \al )} {e^{- \al / u} \over u^{1 + \al }} {{t_i}^{\ga } \exp ( \ga \x_i^{\top} \bbet ) u \pi ( {t_i}^{\ga } \exp ( \ga \x_i^{\top} \bbet ) u) \over C' / [ \{ 1 + \log (1 + t_i ) \} \{ 1 + \log [1 + \log (1 + t_i )] \} ^{1 + c'} ]} du \non \\
&\le {\al ^{\al } e^{- \al } \over \Ga ( \al )} \ga \int_{0}^{\infty } {{t_i}^{\ga } \exp ( \ga \x_i^{\top} \bbet ) \pi ( {t_i}^{\ga } \exp ( \ga \x_i^{\top} \bbet ) u) \over C' / [ \{ 1 + \log (1 + t_i ) \} \{ 1 + \log [1 + \log (1 + t_i )] \} ^{1 + c'} ]} du \non \\
&= {\{ 1 + \log (1 + t_i ) \} \{ 1 + \log [1 + \log (1 + t_i )] \} ^{1 + c'} \over C'} {\al ^{\al } e^{- \al } \over \Ga ( \al )} \ga \non \\
&\le \begin{cases} M ( \min \{ \al , \al ^{1 / 2} \} ) \ga \text{,} & \text{if $i \in \Kc$} \text{,} \\ M \om ( \log \om )^{1 + c'} ( \min \{ \al , \al ^{1 / 2} \} ) \ga \text{,} & \text{if $i \in \Lc$} \text{,} \end{cases} \non 
\end{align}
for some $M > 0$. 
\end{proof}

\begin{proof}[Proof of Lemma \ref{lem:L}]
For part (i), 
\begin{align}
\int_{0}^{1 / (1 + \De )} {\al ^{\al } \over \Ga ( \al )} {e^{- \al / u} \over u^{1 + \al }} du &= \int_{1 + \De }^{\infty } {\al ^{\al } \over \Ga ( \al )} w^{\al - 1} e^{- \al w} dw \non \\
&= \int_{\De }^{\infty } {\al ^{\al } e^{- \al } \over \Ga ( \al )} {1 \over 1 + s} e^{- \al \{ s - \log (1 + s) \} } ds \text{.} \non 
\end{align}
Note that for all $s > 0$, 
\begin{align}
\log (1 + s) &= - \log \Big( 1 - {s \over 1 + s} \Big) = \sum_{k = 1}^{\infty } {1 \over k} \Big( {s \over 1 + s} \Big) ^k \non \\
&\le {s \over 1 + s} + {1 \over 2} \sum_{k = 2}^{\infty } \Big( {s \over 1 + s} \Big) ^k = {s \over 1 + s} + {1 \over 2} {s^2 \over 1 + s} \text{.} \non 
\end{align}
Then 
\begin{align}
\int_{0}^{1 / (1 + \De )} {\al ^{\al } \over \Ga ( \al )} {e^{- \al / u} \over u^{1 + \al }} du &\le \int_{\De }^{\infty } {\al ^{\al } e^{- \al } \over \Ga ( \al )} {1 \over 1 + s} \exp \Big( - {1 \over 2} \al {s^2 \over 1 + s} \Big) ds \non \\
&= \int_{1}^{\infty } {\al ^{\al } e^{- \al } \over \Ga ( \al )} {\De \over 1 + \De t} \exp \Big( - {1 \over 2} \al {\De ^2 t^2 \over 1 + \De t} \Big) dt \non \\
&\le \int_{1}^{\infty } 1( \De t \le 1) {\al ^{\al } e^{- \al } \over \Ga ( \al )} {\De \over 1 + \De t} \exp \Big( - {1 \over 2} \al {\De ^2 t^2 \over 2} \Big) dt \non \\
&\quad + \int_{1}^{\infty } 1( \De t > 1) {\al ^{\al } e^{- \al } \over \Ga ( \al )} {\De \over 2 \De ^{1 / 2} t^{1 / 2}} \exp \Big( - {1 \over 2} \al {\De t \over 2} \Big) dt \non \\
&\le \int_{1}^{\infty } {\al ^{\al - 1 / 2} e^{- \al } \over \Ga ( \al )} \al ^{1 / 2} \De \exp \Big\{ - {t^2 \over 4} ( \al ^{1 / 2} \De )^2 \Big\} dt \non \\
&\quad + \int_{1}^{\infty } {\al ^{\al - 1 / 2} e^{- \al } \over \Ga ( \al )} ( \al \De )^{1 / 2} \exp \Big( - {t \over 4} \al \De \Big) dt \text{.} \non 
\end{align}

For part (ii), 
\begin{align}
\int_{1 + \De }^{K} {\al ^{\al } \over \Ga ( \al )} {e^{- \al / u} \over u^{1 + \al }} du &= \int_{1 / K}^{1 / (1 + \De )} {\al ^{\al } \over \Ga ( \al )} w^{\al - 1} e^{- \al w} dw \non \\
&= \int_{- (K - 1) / K}^{- \De / (1 + \De )} {\al ^{\al } e^{- \al } \over \Ga ( \al )} {1 \over 1 + s} e^{- \al \{ s - \log (1 + s) \} } ds \text{.} \non 
\end{align}
Note that for all $- 1 < s < 0$, 
\begin{align}
- \log (1 + s) &= - \log \{ 1 - (- s) \} = \sum_{k = 1}^{\infty } {1 \over k} |s|^k \ge - s + {1 \over 2} s^2 \text{.} \non 
\end{align}
Then, since 
\begin{align}
\Big( - {K - 1 \over K}, - {\De \over 1 + \De } \Big) &\subset \Big( - \infty , - {\log (1 + \De ) \over 1 + \log (1 + \De )} \Big) \subset \Big( - \infty , - {\log (1 + \De ) \over 1 + \log K} \Big) \text{,} \non 
\end{align}
it follows that 
\begin{align}
&\int_{1 + \De }^{K} {\al ^{\al } \over \Ga ( \al )} {e^{- \al / u} \over u^{1 + \al }} \non \\
&\le K \int_{- (K - 1) / K}^{- \De / (1 + \De )} {\al ^{\al } e^{- \al } \over \Ga ( \al )} \exp \Big( - \al {1 \over 2} s^2 \Big) ds \non \\
&\le K \int_{- \infty }^{- \{ \log (1 + \De ) \} / (1 + \log K)} {\al ^{\al } e^{- \al } \over \Ga ( \al )} \exp \Big( - \al {1 \over 2} s^2 \Big) ds \non \\
&= K \int_{- \infty }^{- 1} {\log (1 + \De ) \over 1 + \log K} {\al ^{\al } e^{- \al } \over \Ga ( \al )} \exp \Big[ - {t^2 \over 2} \al {\{ \log (1 + \De ) \} ^2 \over (1 + \log K)^2} \Big] dt \non \\
&\le K \int_{- \infty }^{- 1} {\al ^{\al - 1 / 2} e^{- \al } \over \Ga ( \al )} |t| \al ^{1 / 2} {\log (1 + \De ) \over 1 + \log K} \exp \Big[ - {|t|^2 \over 2} \al {\{ \log (1 + \De ) \} ^2 \over (1 + \log K)^2} \Big] dt \text{.} \non 
\end{align}

For part (iii), 
\begin{align}
\int_{K}^{\infty } {\al ^{\al } \over \Ga ( \al )} {e^{- \al / u} \over u^{1 + \al }} du &\le \int_{K}^{\infty } {\al ^{\al } \over \Ga ( \al )} {1 \over u^{1 + \al }} du = {\al ^{\al } \over \Ga ( \al )} {1 \over \al } {1 \over K^{\al }} = {\al ^{\al - 1 / 2} e^{- \al } \over \Ga ( \al )} {1 \over \al ^{1 / 2}} {1 \over (K / e)^{\al }} \non \\
&\le {1 \over \al ^{1 / 2}} {1 \over (K / e)^{\al }} \text{.} \non 
\end{align}
Also, 
\begin{align}
\int_{K}^{\infty } {\al ^{\al } \over \Ga ( \al )} {e^{- \al / u} \over u^{1 + \al }} du &\le {\al ^{\al } \over \Ga ( \al )} {1 \over \al } {1 \over K^{\al }} = {\al ^{\al } \over \Ga ( \al + 1)} {1 \over K^{\al }} \text{.} \non 
\end{align}
This completes the proof. 
\end{proof}

\begin{proof}[Proof of Lemma \ref{lem:K_improved}]
If $( \min \{ \al , \al ^{1 / 2} \} ) \ga \le (4 / N) \log (1 + \ga )$, the inequality follows from Lemma \ref{lem:K_new}. 
If $\ga \le 4 / N$, the inequality follows from Lemma \ref{lem:L_2_new}. 
Suppose that $( \min \{ \al , \al ^{1 / 2} \} ) \ga > (4 / N) \log (1 + \ga )$ and that $\ga > 4 / N$. 
Then, by Lemma \ref{lem:L_2_new} and by assumption, 
\begin{align}
&f( \al , \x_i^{\top} \bbet , \ga ; t_i ) \non \\
&\le M_1 \int_{0}^{\infty } {\al ^{\al } \ga \over \Ga ( \al )} {e^{- \al / u} \over u^{1 + \al }} {1 \over 1 + \log \{ 1 + \exp ( \ga | \log t_i + \x_i^{\top} \bbet | ) \min \{ u, 1 / u \} \} } du \non \\
&\le M_1 \int_{0}^{\infty } {\al ^{\al } \ga \over \Ga ( \al )} {e^{- \al / u} \over u^{1 + \al }} {1 \over 1 + \log (1 + e^{N \ga } \min \{ u, 1 / u \} )} du \non \\
&\le M_1 \int_{0}^{1 / e^{N \ga / 2}} {\al ^{\al } \ga \over \Ga ( \al )} {e^{- \al / u} \over u^{1 + \al }} du \non \\
&\quad + M_1 \int_{1 / e^{N \ga / 2}}^{e^{N \ga / 2}} {\al ^{\al } \ga \over \Ga ( \al )} {e^{- \al / u} \over u^{1 + \al }} {1 \over 1 + \log (1 + e^{N \ga / 2} )} du \non \\
&\quad + M_1 \int_{e^{N \ga / 2}}^{\infty } {\al ^{\al } \ga \over \Ga ( \al )} {e^{- \al / u} \over u^{1 + \al }} du \non \\
&\le M_1 \ga \int_{0}^{1 / \{ 1 + (N\ga / 4)^2 \} } {\al ^{\al } \over \Ga ( \al )} {e^{- \al / u} \over u^{1 + \al }} du + M_1 {2 \over N} + M_1 \ga \int_{e^{N \ga / 2}}^{\infty } {\al ^{\al } \over \Ga ( \al )} {e^{- \al / u} \over u^{1 + \al }} du \non 
\end{align}
for some $M_1 > 0$. 
Therefore, by Lemma \ref{lem:L}, 
\begin{align}
&f( \al , \x_i^{\top} \bbet , \ga ; t_i ) \non \\
&\le M_1 \ga \Big[ \int_{1}^{\infty } \al ^{1 / 2} \De \exp \Big\{ - {t^2 \over 4} ( \al ^{1 / 2} \De )^2 \Big\} dt \non \\
&\quad + \int_{1}^{\infty } ( \al \De )^{1 / 2} \exp \Big( - {t \over 4} \al \De \Big) dt \Big] \Big| _{\De = (N
\ga / 4)^2} + M_1 {2 \over N} + M_2 \ga {1 \over e^{N \al \ga / 4}} \non \\
&\le M_3 \Big[ \ga \Big\{ \exp \Big( - {1 \over 12} {\De _1}^2 \Big) \int_{1}^{\infty } \exp \Big( - {t^2 \over 12} {\De _1}^2 \Big) dt \Big\} \Big| _{\De _1 = \al ^{1 / 2} (N \ga / 4)^2} \non \\
&\quad + \ga \Big\{ \exp \Big( - {1 \over 12} \De _2 \Big) \int_{1}^{\infty } \exp \Big( - {t \over 12} \De _2 \Big) dt \Big\} \Big| _{\De _2 = \al (N \ga / 4)^2} + 1 + \ga {1 \over e^{N \al \ga / 4}} \Big] \non \\
&\le M_3 \Big[ \ga \Big\{ \exp \Big( - {1 \over 12} {\De _1}^2 \Big) \int_{1}^{\infty } \exp \Big( - {t^2 \over 12} {\De _1}^2 \Big) dt \Big\} \Big| _{\De _1 = (N \ga / 4) \log (1 + \ga )} + 1 \non \\
&\quad + \ga \Big\{ \exp \Big( - {1 \over 12} \De _2 \Big) \int_{1}^{\infty } \exp \Big( - {t \over 12} \De _2 \Big) dt \Big\} \Big| _{\De _2 = (N \ga / 4) \log (1 + \ga )} + \ga {1 \over e^{\log (1 + \ga )}} \Big] \le M_4 \non 
\end{align}
for some $M_2 , M_3 , M_4 > 0$. 
\end{proof}

\begin{proof}[Proof of Lemma \ref{lem:L_improved}]
By Lemma \ref{lem:L_2_new} and by assumption, 
\begin{align}
&f( \al , \x_i^{\top} \bbet , \ga ; t_i ) \non \\
&\le M_1 \int_{0}^{\infty } \Big\{ {\al ^{\al } \ga \over \Ga ( \al )} {e^{- \al / u} \over u^{1 + \al }} {1 + \log (1 + t_i ) \over 1 + \log \{ 1 + ( {t_i}^{\ep } \min \{ u^{1 / \ga } , 1 / u^{1 / \ga } \} )^{\ga } \} } \non \\
&\quad \times {[1 + \log \{ 1 + \log (1 + t_i ) \} ]^{1 + c} \over (1 + \log [1 + \log \{ 1 + ( {t_i}^{\ep } \min \{ u^{1 / \ga } , 1 / u^{1 / \ga } \} )^{\ga } \} ])^{1 + c}} \Big\} du \non \\
&\le M_1 \Big\{ \ga \int_{0}^{1 / {t_i}^{\ga \ep / 2}} {\al ^{\al } \over \Ga ( \al )} {e^{- \al / u} \over u^{1 + \al }} \{ 1 + \log (1 + t_i ) \} [1 + \log \{ 1 + \log (1 + t_i ) \} ]^{1 + c} du \non \\
&\quad + \int_{1 / {t_i}^{\ga \ep / 2}}^{{t_i}^{\ga \ep / 2}} {\al ^{\al } \over \Ga ( \al )} {e^{- \al / u} \over u^{1 + \al }} {\ga \{ 1 + \log (1 + t_i ) \} \over 1 + \log \{ 1 + ( {t_i}^{\ep / 2} )^{\ga } \} } {[1 + \log \{ 1 + \log (1 + t_i ) \} ]^{1 + c} \over (1 + \log [1 + \log \{ 1 + ( {t_i}^{\ep / 2} )^{\ga } \} ])^{1 + c}} du \non \\
&\quad + \ga \int_{{t_i}^{\ga \ep / 2}}^{\infty } {\al ^{\al } \over \Ga ( \al )} {e^{- \al / u} \over u^{1 + \al }} \{ 1 + \log (1 + t_i ) \} [1 + \log \{ 1 + \log (1 + t_i ) \} ]^{1 + c} du \Big\} \non \\
&\le M_2 \Big[ \{ 1 + \log (1 + 1 / \ga ) \} ^{1 + c} \int_{1 / {t_i}^{\ga \ep / 2}}^{{t_i}^{\ga \ep / 2}} {\al ^{\al } \over \Ga ( \al )} {e^{- \al / u} \over u^{1 + \al }} du \non \\
&\quad + \om ( \log \om )^{1 + c} \ga \int_{0}^{1 / {t_i}^{\ga \ep / 2}} {\al ^{\al } \over \Ga ( \al )} {e^{- \al / u} \over u^{1 + \al }} du + \om ( \log \om )^{1 + c} \ga \int_{{t_i}^{\ga \ep / 2}}^{\infty } {\al ^{\al } \over \Ga ( \al )} {e^{- \al / u} \over u^{1 + \al }} du \Big] \non 
\end{align}
for some $M_1 , M_2 > 0$. 
Clearly, 
\begin{align}
\{ 1 + \log (1 + 1 / \ga ) \} ^{1 + c} \int_{1 / {t_i}^{\ga \ep / 2}}^{{t_i}^{\ga \ep / 2}} {\al ^{\al } \over \Ga ( \al )} {e^{- \al / u} \over u^{1 + \al }} du \le \{ 1 + \log (1 + 1 / \ga ) \} ^{1 + c} \text{.} \non 
\end{align}

By part (i) of Lemma \ref{lem:L}, 
\begin{align}
&\int_{0}^{1 / {t_i}^{\ga \ep / 2}} {\al ^{\al } \over \Ga ( \al )} {e^{- \al / u} \over u^{1 + \al }} du \le \int_{0}^{1 / \{ 1 + ( \ga \ep / 2) \log t_i \} } {\al ^{\al } \over \Ga ( \al )} {e^{- \al / u} \over u^{1 + \al }} du \non \\
&\le \Big[ \int_{1}^{\infty } \al ^{1 / 2} \De \exp \Big\{ - {t^2 \over 4} ( \al ^{1 / 2} \De )^2 \Big\} dt + \int_{1}^{\infty } ( \al \De )^{1 / 2} \exp \Big( - {t \over 4} \al \De \Big) dt \Big] \Big| _{\De = ( \ga \ep / 2) \log t_i} \non \\
&\le M_3 \Big[ \exp \Big\{ - {1 \over 12} ( \al ^{1 / 2} \De )^2 \Big\} \int_{1}^{\infty } \exp \Big\{ - {t^2 \over 12} ( \al ^{1 / 2} \De )^2 \Big\} dt \non \\
&\quad + \exp \Big( - {1 \over 12} \al \De \Big) \int_{1}^{\infty } \exp \Big( - {t \over 12} \al \De \Big) dt \Big] \Big| _{\De = ( \ga \ep / 2) \log t_i} \non \\
&\le M_3 \Big\{ \exp \Big( - {\De ^2 \over 12} \Big) \int_{1}^{\infty } \exp \Big( - {t^2 \over 12} \De ^2 \Big) dt \non \\
&\quad + \exp \Big( - {\De \over 12} \Big) \int_{1}^{\infty } \exp \Big( - {t \over 12} \De \Big) dt \Big\} \Big| _{\De = ( \min \{ \al , \al ^{1 / 2} \} ) ( \ga \ep / 2) \log t_i} \non \\
&\le M_3 \Big\{ \exp \Big( - {\De ^2 \over 12} \Big) \int_{1}^{\infty } \exp \Big( - {t^2 \over 12} \De ^2 \Big) dt \non \\
&\quad + \exp \Big( - {\De \over 12} \Big) \int_{1}^{\infty } \exp \Big( - {t \over 12} \De \Big) dt \Big\} \Big| _{\De = 2 \om ^{\eta }} \non \\
&\le M_4 \Big\{ \exp \Big( - {\om ^{2 \eta } \over 3} \Big) + \exp \Big( - {\om ^{\eta } \over 6} \Big) \Big\} \non 
\end{align}
for some $M_3 , M_4 > 0$. 
Similarly, 
\begin{align}
&\int_{0}^{1 / {t_i}^{\ga \ep / 2}} {\al ^{\al } \over \Ga ( \al )} {e^{- \al / u} \over u^{1 + \al }} du \le \int_{0}^{1 / \{ 1 + ( \ga \ep / 4)^2 ( \log t_i )^2 \} } {\al ^{\al } \over \Ga ( \al )} {e^{- \al / u} \over u^{1 + \al }} du \non \\
&\le \Big[ \int_{1}^{\infty } \al ^{1 / 2} \De \exp \Big\{ - {t^2 \over 4} ( \al ^{1 / 2} \De )^2 \Big\} dt + \int_{1}^{\infty } ( \al \De )^{1 / 2} \exp \Big( - {t \over 4} \al \De \Big) dt \Big] \Big| _{\De = ( \ga \ep / 4)^2 ( \log t_i )^2} \non \\
&\le M_5 \Big[ \exp \Big\{ - {1 \over 12} ( \al ^{1 / 2} \De )^2 \Big\} \int_{1}^{\infty } \exp \Big\{ - {t^2 \over 12} ( \al ^{1 / 2} \De )^2 \Big\} dt \non \\
&\quad + \exp \Big( - {1 \over 12} \al \De \Big) \int_{1}^{\infty } \exp \Big( - {t \over 12} \al \De \Big) dt \Big] \Big| _{\De = ( \ga \ep / 4)^2 ( \log t_i )^2} \non \\
&\le M_5 \Big\{ \exp \Big( - {\De ^2 \over 12} \Big) \int_{1}^{\infty } \exp \Big( - {t^2 \over 12} \De ^2 \Big) dt \non \\
&\quad + \exp \Big( - {\De \over 12} \Big) \int_{1}^{\infty } \exp \Big( - {t \over 12} \De \Big) dt \Big\} \Big| _{\De = ( \min \{ \al , \al ^{1 / 2} \} ) ( \ga \ep / 4)^2 ( \log t_i )^2} \non \\
&\le M_5 \Big\{ \exp \Big( - {\De ^2 \over 12} \Big) \int_{1}^{\infty } \exp \Big( - {t^2 \over 12} \De ^2 \Big) dt \non \\
&\quad + \exp \Big( - {\De \over 12} \Big) \int_{1}^{\infty } \exp \Big( - {t \over 12} \De \Big) dt \Big\} \Big| _{\De = \om ^{\eta } ( \ga \ep / 4) \log t_i} \non \\
&\le M_6 \Big[ \exp \Big\{ - {\om ^{2 \eta } \over 12} ( \ga \ep / 4)^2 ( \log t_i )^2 \Big\} + \exp \Big\{ - {\om ^{\eta } \over 12} ( \ga \ep / 4) \log t_i \Big\} \Big] \non 
\end{align}
for some $M_5 , M_6 > 0$ if $\ga \ge 1$. 
Therefore, 
\begin{align}
&\om ( \log \om )^{1 + c} \ga \int_{0}^{1 / {t_i}^{\ga \ep / 2}} {\al ^{\al } \over \Ga ( \al )} {e^{- \al / u} \over u^{1 + \al }} du \le M_7 \non 
\end{align}
for some $M_7 > 0$. 

Suppose first that $\ga \le 1$. 
Then, by parts (ii) and (iii) of Lemma \ref{lem:L}, 
\begin{align}
&\int_{{t_i}^{\ga \ep / 2}}^{\infty } {\al ^{\al } \over \Ga ( \al )} {e^{- \al / u} \over u^{1 + \al }} du \non \\
&\le 1( {t_i}^{\ga \ep / 2} \le e + 1) \int_{{t_i}^{\ga \ep / 2}}^{e + 1} {\al ^{\al } \over \Ga ( \al )} {e^{- \al / u} \over u^{1 + \al }} du + \int_{\max \{ {t_i}^{\ga \ep / 2} , e + 1 \} }^{\infty } {\al ^{\al } \over \Ga ( \al )} {e^{- \al / u} \over u^{1 + \al }} du \non \\
&\le \Big( 1( {t_i}^{\ga \ep / 2} \le e + 1) \non \\
&\quad \times (e + 1) \int_{- \infty }^{- 1} |t| \al ^{1 / 2} {( \log t_i ) \ga \ep / 2 \over 1 + \log (e + 1)} \exp \Big[ - {|t|^2 \over 2} \al {\{ ( \log t_i ) \ga \ep / 2 \} ^2 \over \{ 1 + \log (e + 1) \} ^2} \Big] dt \non \\
&\quad + \min \Big\{ {\al ^{\al } \over \Ga ( \al + 1)} {1 \over ( \max \{ {t_i}^{\ga \ep / 2} , e + 1 \} )^{\al }}, {1 \over \al ^{1 / 2}} {1 \over \{ ( \max \{ {t_i}^{\ga \ep / 2} , e + 1 \} ) / e \} ^{\al }} \Big\} \Big) \non \\
&\le M_8 \Big( (e + 1) \exp \Big[ - {1 \over 6} \al {\{ ( \log t_i ) \ga \ep / 2 \} ^2 \over \{ 1 + \log (e + 1) \} ^2} \Big] \int_{- \infty }^{- 1} \exp \Big[ - {|t|^2 \over 6} \al {\{ ( \log t_i ) \ga \ep / 2 \} ^2 \over \{ 1 + \log (e + 1) \} ^2} \Big] dt \non \\
&\quad + 1( \al \le 1) {1 \over ( {t_i}^{\ga \ep / 2} )^{\al }} + 1( \al > 1) {1 \over \max \{ \exp [ \al \{ ( \log t_i ) \ga \ep / 2 - 1 \} ], (1 + 1 / e)^{\al } \} } \Big) \non 
\end{align}
for some $M_8 > 0$. 
Since 
\begin{align}
&\max \{ \exp [ \al \{ ( \log t_i ) \ga \ep / 2 - 1 \} ], (1 + 1 / e)^{\al } \} \non \\
&\ge \begin{cases} \exp \{ \al ( \log t_i ) \ga \ep / 4 \} \ge \exp ( \om ^{\eta } ) \text{,} & \text{if $( \log t_i ) \ga \ep / 4 \ge 1$} \text{,} \\ (1 + 1 / e)^{\om ^{\eta } / \{ ( \log t_i ) \ga \ep / 4 \} } \ge (1 + 1 / e)^{\om ^{\eta }} \text{,} & \text{if $( \log t_i ) \ga \ep / 4 < 1$} \text{,} \end{cases} \non 
\end{align}
it follows that 
\begin{align}
&\om ( \log \om )^{1 + c} \ga \int_{{t_i}^{\ga \ep / 2}}^{\infty } {\al ^{\al } \over \Ga ( \al )} {e^{- \al / u} \over u^{1 + \al }} du \non \\
&\le M_8 \om ( \log \om )^{1 + c} \ga \non \\
&\quad \times \Big( (e + 1) \exp \Big[ - {1 \over 6} \al {\{ ( \log t_i ) \ga \ep / 2 \} ^2 \over \{ 1 + \log (e + 1) \} ^2} \Big] \int_{- \infty }^{- 1} \exp \Big[ - {|t|^2 \over 6} \al {\{ ( \log t_i ) \ga \ep / 2 \} ^2 \over \{ 1 + \log (e + 1) \} ^2} \Big] dt \non \\
&\quad + 1( \al \le 1) {1 \over ( {t_i}^{\ga \ep / 2} )^{\al }} + 1( \al > 1) {1 \over (1 + 1 / e)^{\om ^{\eta }}} \Big) \non \\
&\le M_9 \om ( \log \om )^{1 + c} \ga \Big( \exp \Big[ - {1 \over 6} {(2 \om ^{\eta } )^2 \over \{ 1 + \log (e + 1) \} ^2} \Big] + {1 \over \exp (2 \om ^{\eta } )} + {1 \over (1 + 1 / e)^{\om ^{\eta }}} \Big) \le M_{10} \non 
\end{align}
for some $M_9 , M_{10} > 0$. 
Next, suppose that $\ga > 1$. 
Then, by part (iii) of Lemma \ref{lem:L}, 
\begin{align}
\int_{{t_i}^{\ga \ep / 2}}^{\infty } {\al ^{\al } \over \Ga ( \al )} {e^{- \al / u} \over u^{1 + \al }} du &\le M_{11} \Big\{ 1( \al \le 1) {1 \over ( {t_i}^{\ga \ep / 2} )^{\al }} + 1( \al > 1) {1 \over ( {t_i}^{\ga \ep / 2} / e)^{\al }} \Big\} \non \\
&\le M_{11} {1 \over {t_i}^{\al \ga \ep / 4}} \le M_{11} {1 \over \exp ( \om ^{\eta } / 2)} {1 \over \exp \{ \al \ga ( \ep / 8) \log t_i \} } \text{.} \non 
\end{align}
Therefore, if $\al \ga \ge \log (1 + \ga )$ or $\ga \le \exp ( \om ^{\eta } / 4)$, 
\begin{align}
&\om ( \log \om )^{1 + c} \ga \int_{{t_i}^{\ga \ep / 2}}^{\infty } {\al ^{\al } \over \Ga ( \al )} {e^{- \al / u} \over u^{1 + \al }} du \non \\
&\le M_{11} {\om ( \log \om )^{1 + c} \over \exp ( \om ^{\eta } / 4)} {\ga \over \exp ( \om ^{\eta } / 4) \exp \{ \al \ga ( \ep / 8) \log t_i \} } \le M_{12} \non 
\end{align}
for some $M_{12} > 0$. 
This completes the proof. 
\end{proof}

\section{Importance of doubly log-adjustment: non-robustness of singly log-adjusted priors}
\label{subsec:supp-nonrobust}

Here we discuss the non-robustness of the prior without doubly log-adjustment. 
Under the setting of Section \ref{sec:3} of the main text, instead of using the DLH distribution, we consider the following local prior: 
\begin{align}
\pi _{(1)} ( \la _i ) \equiv  \pi _{(1)} ( \la _i | d) 
= {d \over 1 + \la _i} {1 \over \{ 1 + \log (1 + \la _i ) \} ^{1 + d}} \non 
\end{align}
for all $i = 1, \dots , n$ for $d > 0$.
The difference from the DLH distribution is that the above density holds only a single log-adjustment term while the tail decay is in the same order as DLH without log-factors. 
We show that under the above local prior, the posterior robustness will not hold in general. 
More specifically, the posterior distribution is expected to be approximately proportional to 
\begin{align}
&{1 \over \ga ^{| \Lc | d}} \pi ( \al , \bbe , \ga ) \prod_{i \in \Kc } p( t_i | \al , \bbe , \ga ) \text{,} \non 
\end{align}
which is not proportional to $p( \al , \bbe , \ga | \{ t_i | i \in \Kc \} )$. 

To see this, fix $i = 1, \dots , n$. 
Then 
\begin{align}
{p( t_i | \al , \bbe , \ga ) \over \pi _{(1)} ( t_i )} %
&= \int_{0}^{\infty } {\al ^{\al } \ga \over \Ga ( \al )} {\{ \exp ( \x_i^{\top} \bbe ) \} ^{\al } \over {\la _i}^{\al }} {t_i}^{\ga \al } e^{- \al \{ \exp ( \x_i^{\top} \bbe ) / \la _i \} {t_i}^{\ga }} {\pi _{(1)} ( \la _i ) \over t_i \pi _{(1)} ( t_i )} d{\la _i} \non \\
&= \int_{0}^{\infty } {\al ^{\al } \ga \over \Ga ( \al )} {e^{- \al / u} \over u^{1 + \al }} {\pi _{(1)} ( {t_i}^{\ga } \exp ( \x_i^{\top} \bbe ) u ) \over t_i \pi _{(1)} ( t_i ) / \{ {t_i}^{\ga } \exp ( \x_i^{\top} \bbe ) u \} } du \text{.} \non 
\end{align}
Note that if $t_i > 1$, then 
\begin{align}
&{\pi _{(1)} ( {t_i}^{\ga } \exp ( \x_i^{\top} \bbe ) u ) \over t_i \pi _{(1)} ( t_i ) / \{ {t_i}^{\ga } \exp ( \x_i^{\top} \bbe ) u \} } \non \\
&= {(1 + t_i ) {t_i}^{\ga } \exp ( \x_i^{\top} \bbe ) u \over t_i \{ 1 + {t_i}^{\ga } \exp ( \x_i^{\top} \bbe ) u \} } { \{ 1 + \log (1 + t_i ) \} ^{1 + d} \over [1 + \log \{ 1 + {t_i}^{\ga } \exp ( \x_i^{\top} \bbe ) u \} ]^{1 + d}} \non \\
&\le {( t_i + t_i ) {t_i}^{\ga } \exp ( \x_i^{\top} \bbe ) u \over t_i \{ 1 + {t_i}^{\ga } \exp ( \x_i^{\top} \bbe ) u \} } { \{ 1 + \log ( t_i + t_i ) \} ^{1 + d} \over \{1 + \log (0 + {t_i}^{\ga } ) \} ^{1 + d}} {\{1 + \log (1 + {t_i}^{\ga } ) \} ^{1 + d} \over [1 + \log \{ 1 + {t_i}^{\ga } \exp ( \x_i^{\top} \bbe ) u \} ]^{1 + d}} \non \\
&\le 2 \Big\{ (1 + \log 2) + {1 \over \ga } \Big\} ^{1 + d} \Big[ 1 + \log \Big\{ 1 + {1 \over \exp ( \x_i^{\top} \bbe ) u} \Big\} \Big]^{1 + d} \text{,} \non 
\end{align}
where the second inequality follows from Lemma \ref{lem:log}, for all $u \in (0, \infty )$. 
Also, note that by Lemma \ref{lem:log}, 
\begin{align}
&{1 + \log (1 + t_i ) \over 1 + \log \{ 1 + {t_i}^{\ga } \exp ( \x_i^{\top} \bbe ) u \} } \non \\
&= {1 + \log t_i \over 1 + \log ( {t_i}^{\ga _i} )} {\{ 1 + \log (1 + t_i ) \} / (1 + \log t_i ) \over \{ 1 + \log (1 + {t_i}^{\ga _i} ) \} / \{ 1 + \log ( {t_i}^{\ga _i} ) \} } {1 + \log (1 + {t_i}^{\ga _i} ) \over 1 + \log \{ 1 + {t_i}^{\ga } \exp ( \x_i^{\top} \bbe ) u \} } \to {1 \over \ga } \non 
\end{align}
as $t_i \to \infty $ for all $u \in (0, \infty )$. 
Then 
\begin{align}
\lim_{t_i \to \infty } {p( t_i | \al , \bbe , \ga ) \over \pi _{(1)} ( t_i )} &= \int_{0}^{\infty } {\al ^{\al } \ga \over \Ga ( \al )} {e^{- \al / u} \over u^{1 + \al }} \Big[ \lim_{t_i \to \infty } {\pi _{(1)} ( {t_i}^{\ga } \exp ( \x_i^{\top} \bbe ) u ; d) \over t_i \pi _{(1)} ( t_i ; d) / \{ {t_i}^{\ga } \exp ( \x_i^{\top} \bbe ) u \} } \Big] du = {1 \over \ga ^d} \text{.} \non 
\end{align}
by the dominated convergence theorem.

\section{On a log-Student $t$ model}
\label{subsec:supp-logT} 
We discuss non-robustness of the log-Student $t$ model \citep[e.g.][]{VallejosSteel2015}.
We consider the log-Student $t$ model for survival outcome $t_i$, where the density is given by 
\begin{align}
&{\Ga (1 / 2 + \nu / 2) ( \nu / 2)^{\nu / 2} \over \sqrt{2 \pi } \Ga ( \nu / 2)} {1 \over \si } {1 \over t_i} {1 \over \{ ( \log t_i - \x_i^{\top} \bbe )^2 / (2 \si ^2 ) + \nu / 2 \} ^{1 / 2 + \nu / 2}},
\end{align}
for $i=1,\ldots,n$ and $\nu>0$. 
Let $p( \bbe , \si )$ be a joint prior for $\bbe$ and $\si$. 
Then for all $i = 1, \dots , n$, we have 
\begin{align}
{p( t_i | \bbe , \si ) \over \{ p( t_i | \bbe , \si ) \} |_{( \bbe , \si ) = ( \bm{0} ^{(p)} , 1)}} 
&= {1 \over \si } {\{ 1 + ( \log t_i )^2 / \nu \} ^{1 / 2 + \nu / 2} \over \{ 1 + ( \log t_i - \x_i^{\top} \bbe )^2 / ( \nu \si ^2 ) \} ^{1 / 2 + \nu / 2}} \to 
\si ^{\nu } \non 
\end{align}
as $t_i \to \infty $. 
Therefore, if $\Kc $ and $\Lc $ are the sets of indices of non-outliers and outliers, the posterior distribution of $( \bbe , \si )$ will be approximately proportional to 
\begin{align}
&\si ^{| \Lc | \nu } p( \bbe , \si ) \prod_{i \in \Kc } p( t_i | \bbe , \si ) \text{,} \non 
\end{align}
which is not proportional to the target density in the present setting, namely, $p( \bbe , \si | \{ t_i | i \in \Kc \} )$. 
This is parallel to the result in the Student $t$ case of \cite{gagnon2023theoretical}. 

\end{document}